\documentclass[11pt]{article}

\usepackage[T1]{fontenc}
\usepackage[utf8]{inputenc} 
\usepackage{lmodern} 
\usepackage{microtype}

\usepackage[margin=1in]{geometry}
\usepackage{setspace}
\usepackage{amsfonts,amsmath,amssymb,amsthm,mathtools,bm}
\mathtoolsset{showonlyrefs}
\numberwithin{equation}{section}
\makeatletter
\renewcommand*{\MT@newlabel}[1]{%
  \global\@namedef{MT_r_#1}{}%
  \global\@namedef{MT_r_{#1}}{}%
}
\makeatother

\usepackage{graphicx}
\usepackage{enumitem}
\setlist{nosep}

\usepackage[round]{natbib}

\usepackage[dvipsnames]{xcolor}
\usepackage{hyperref}
\hypersetup{
  colorlinks=true,
  linkcolor=MidnightBlue,
  citecolor=OliveGreen,
  urlcolor=RoyalBlue,
  pdfauthor={},
  pdftitle={}
}

\usepackage[nameinlink,capitalize,noabbrev]{cleveref}

\usepackage[title]{appendix}

\theoremstyle{plain}
\newtheorem{assumption}{Assumption}
\newtheorem{theorem}{Theorem}[section]
\newtheorem{lemma}[theorem]{Lemma}
\newtheorem{proposition}[theorem]{Proposition}

\theoremstyle{definition}

\theoremstyle{remark}

\newcommand{\mc}{\mathcal}
\newcommand{\mbb}{\mathbb}

\newcommand{\ve}{\varepsilon}

\DeclareMathOperator*{\argmin}{arg\,min}

\DeclareMathOperator{\Var}{Var}
\DeclareMathOperator{\Cov}{Cov}

\begin{document}

\title{On prediction-powered inference for quantile regression via convolution smoothing}

\author{
  Shota Takeishi, Jimin Ding and Xuming He\thanks{Xuming He acknowledges support from NSF Grant DMS-2345035.} \\
  Department of Statistics and Data Science, Washington University in St.\ Louis \\
  \texttt{\{shota,jmding,hex\}@wustl.edu}
}

\date{\today} 

\maketitle

\begin{abstract}
This paper studies quantile regression in a data-limited setting where the gold-standard outcome is available only for a limited number of observations, whereas a surrogate outcome is widely available. 
Such settings are becoming increasingly common with the availability of low-cost predictions from modern AI, motivating a growing line of research on ``prediction-powered inference,'' for improved statistical inference.
Naively extending this framework to quantile regression, however, raises two challenges: computational difficulties due to the discontinuity of the subgradient, and overly conservative confidence intervals.
To address these issues, we propose a convolution-based smoothing of the check-loss objective and develop two variants of the estimator.
The proposed estimators are computationally tractable, and our numerical studies show that they mitigate overcoverage.
As a theoretical contribution, we establish the asymptotic distributions of the proposed estimators under a possibly misspecified linear quantile regression model.
We further propose an ensemble of the two estimators and illustrate the proposed methods through simulations and an application to a local housing dataset.
\end{abstract}

\section{Introduction}\label{sec:intro}

Since the seminal work by \cite{KoenkerBassett1978ecta},
quantile regression has been one of the indispensable tools in statistical analysis.
Looking beyond the conditional mean to conditional quantiles,
it captures the heterogeneous relationships between a response and regressors that 
standard linear regression cannot.
Motivated by diverse applications,
quantile regression has been extended along many directions, including
high-dimensional settings \citep{BelloniChernozhukov2011aos},
post-selection inference \citep{WangPanigrahiHe2025aos},
censored outcomes \citep{Portnoy2003jasa}, and 
longitudinal data \citep{Koenker2004jmva}.

This paper develops new quantile regression methods 
for another scientifically relevant setting—one with limited availability of gold‑standard data.
We focus on situations where the outcomes of interest are costly or time-consuming to obtain and therefore scarce,
while their predictions or surrogates are abundant.
Our illustrative example uses housing data from selected municipalities in the St.~Louis~area in the United States, and examines how school-district quality is associated with different conditional quantiles of  single-family home prices.  Actual housing prices provide the gold-standard outcome, but they are observed only for homes sold during the study period. In contrast, appraisal values are available for a much larger set of properties, including off-market homes, and contain useful information about market values.
More broadly, such settings are becoming increasingly common across scientific domains as 
powerful AI tools can often produce 
accurate predictions at low cost.

To boost the efficiency of statistical inference in such data-limited settings,
the literature has developed various methods for leveraging a large volume of predicted outcomes.
One of the earliest contributions by \cite{RobinsRotnitzkyZhao1994jasa} proposes augmenting the original
estimating equation to obtain semiparametric efficient estimators.
Despite its theoretical appeal, the augmentation term requires either correct parametric specification or
sufficiently accurate nonparametric nuisance estimation to achieve the desired efficiency, which may be difficult to verify or satisfy in practice.

A line of recent papers has focused on
more pragmatic approaches, here called ``prediction-powered inference'' \citep{Angelopoulos_etal2023science}, which require few modeling assumptions or nonparametric nuisance function estimation,
yet still can deliver some efficiency gain.
Depending on the stage at which augmentation for efficiency is applied,
the prediction-powered inference can be divided into two subtypes: score-debiasing (SD) approach and 
predict-then-debias (PTD) approach \citep{KlugerLuZrnicWangBates2025arxiv}.
The former approach modifies the loss function (or the corresponding score equation) based on gold-standard data by
adding an augmentation term built from losses or scores that use predicted outcomes.
Starting with \cite{Angelopoulos_etal2023science}, multiple variants have been proposed
\citep[e.g.,][]{AngelopoulosDuchiZrnic2023arxiv, MiaoMiaoWuZhaoLu2025jmlr, GanLiangZou2024anzjs,JiLeiZrnic2025arxiv}.
In contrast, the PTD approach first estimates the parameter using only gold-standard data.
It then augments this ``naive'' estimator with a correction term based on estimators that use predicted outcomes.
Pioneered by \cite{ChenChen2000jrssb}, this approach has attracted renewed interest in the recent literature
\citep[see, e.g.,][]{KlugerLuZrnicWangBates2025arxiv, Gronsbell_etal2024arxiv, McCaw_etal2024nature_genetics,
Miao_etal2024nature_genetics}.

In principle, both SD and PTD
approaches can be applied to quantile regression; however, this extension poses some practical challenges.
First, in most existing SD approaches, naively augmenting the check loss can yield a nonconvex, nondifferentiable
objective, while augmenting the score (sub-gradient of the check loss) leads to discontinuous equation; both are difficult to optimize or solve with guaranteed accuracy,
especially when the covariate dimension is not low. This difficulty is further exacerbated when using the bootstrap, a standard choice for quantile regression inference.
By contrast, the PTD approach is straightforward to implement for quantile regression using off-the-shelf software: it only requires
fitting linear quantile regression separately on the gold-standard data and on the data with predicted outcomes. Nevertheless, 
\cite{KlugerLuZrnicWangBates2025arxiv} report overcoverage of its confidence intervals in their numerical study,
which we also observe in our simulations (Section \ref{sec:simulation}).

To address these practical limitations, we develop smoothed estimators for both SD and PTD
quantile regression.
Inspired by \cite{FernandesGuerreHorta2021jbes} and \cite{HePanTanZhou2023joe},
we approximate the nondifferentiable check loss with a smooth, convex surrogate via convolution;
accordingly, the resulting subgradient becomes smooth as well.
Building on this approximation, we propose the score-debiasing convolution smoothed estimator (SD-CSE)  as the solution to the augmented smoothed subgradient condition, and the predict-then-debias convolution smoothed estimator (PTD-CSE), which augments the smoothed quantile regression estimator based on gold-standard data with a correction term.
For both estimators, we choose the augmentation terms so that they are more efficient than
the gold-standard-only estimators when predictions are reasonably accurate, yet at least as efficient
even when predictions are poor.
With smoothing, our SD-CSE alleviates the computational bottleneck noted above for the check-loss-based SD estimator.
Furthermore, through simulation studies, we find that our PTD-CSE mitigates the overcoverage observed with the check-loss-based PTD estimator.
We also propose an ensemble of these two estimators to further improve efficiency.

On the theoretical side, we derive the asymptotic distributions of the SD-CSE and PTD-CSE, constituting a nontrivial extension of the literature on prediction-powered inference and on quantile regression.
The smoothed estimating equation for the SD-CSE depends on a smoothing parameter that varies with sample size,
whereas existing proof strategies for prediction-powered inference \citep[e.g.][]{MiaoMiaoWuZhaoLu2025jmlr}
assume that the equation’s form does not change with sample size.
This necessitates a tailored argument in which the convergence rate of the smoothing parameter plays a key role.
Although \cite{FernandesGuerreHorta2021jbes} and \cite{HePanTanZhou2023joe}
conduct extensive asymptotic analyses of smoothed quantile regression estimators,
their arguments cannot be directly imported to our setting for two reasons.
First, our SD estimator is not defined as the minimizer of a convex loss, as in the works above;
for this reason, we develop a new argument to establish its convergence rate, a critical step toward proving asymptotic normality.
Second, for both the SD and PTD estimators, we allow for misspecification of the linear quantile regression model, which prior works do not consider.
This misspecification requires a new argument to assess the smoothing bias.

The remainder of the paper is organized as follows.
Section \ref{sec:method} introduces the problem setup and defines the SD-CSE and PTD-CSE, along with their ensemble.
Section \ref{sec:theory} establishes the asymptotic distributions of the SD-CSE and PTD-CSE after introducing the required assumptions.
Section \ref{sec:simulation} evaluates finite-sample performance via simulations, and Section \ref{sec:data_analysis} presents a real-world data analysis.
All technical details are provided in the appendix.

\section{Model and estimators}\label{sec:method}
For a fixed quantile level $\tau \in (0, 1)$, we consider inference on a finite-dimensional parameter $\beta_0(\tau) \in \mbb R^p$,
defined as the solution to the population minimization problem for a random response $Y \in \mbb R$ and covariates $X \in \mbb R^p$:
\begin{equation}
    \beta_0(\tau) := \argmin_{\beta \in \mbb R^p} \mbb E[\rho_{\tau} (Y - X^{\top}\beta) - \rho_{\tau}(Y)],
\end{equation}
where $\rho_{\tau}(u) := u \{ \tau - \mbb I(u < 0) \}$ is the check loss. The centering term $-\rho_{\tau}(Y)$ is included only to ensure that the the expectation is well-defined, without imposing finite-moment condition on $Y$. See the discussion following Assumption~\ref{as:identification} for detail.
If the conditional $\tau$-th quantile of $Y$ given $X$, denoted by $Q_{\tau}(Y \mid X)$, satisfies the linear quantile regression specification $Q_{\tau}(Y \mid X) = X^{\top}\beta_*(\tau)$,
then $\beta_0(\tau)$ coincides with $\beta_* (\tau)$.
However, we allow for potential misspecification $Q_{\tau}(Y|X) \neq X^{\top}\beta_{0}(\tau)$; see \cite{AngristChernozhukovFernandezval2006ecta} for an interpretation of $\beta_0(\tau)$ under misspecification of the linear quantile regression model.

In order to make inference about $\beta_0(\tau)$, suppose we have access to two independent random samples:
 $\{(Y_i, \hat Y_i, X_i)\}_{i \in \mathcal{I}_{\text{lab}}}$,
and $\{ (\hat Y_i, X_i) \}_{i \in \mathcal{I}_{\text{unlab}}}$, where each $(Y_i, X_i)$ is a copy of $(Y, X)$.
$\mathcal{I}_{\text{lab}} := \{1, \dots, n \}$ and $\mathcal{I}_{\text{unlab}} := \{n+1, \dots, n+N \}$
denote index sets for the labeled and unlabeled samples, respectively.
The labeled (gold-standard) dataset $\{(Y_i, \hat Y_i, X_i)\}_{i \in \mathcal{I}_{\text{lab}}}$ consists of relatively small size $n$, in which $Y_i$ is observed.
In contrast, the unlabeled dataset $\{ (\hat Y_i, X_i) \}_{i \in \mathcal{I}_{\text{unlab}}}$ is of a larger size $N$, in which $Y_i$ is missing.
Across both datasets, $X_i$ and $\hat Y_i$ are always available, where $\hat Y_i$ is a surrogate for (or prediction of) $Y_i$.
In recent works \citep[e.g.,][]{Angelopoulos_etal2023science,MiaoMiaoWuZhaoLu2025jmlr}, $\hat Y_i$ is often treated as the output from a pre-trained prediction model.
Namely, $\hat Y_i = \hat f (X_i, Z_i)$, where $Z_i$
denotes auxiliary features–often unstructured or high-dimensional (e.g., text or image data), and 
$\hat f(\cdot)$ is a predictive model trained on a separate dataset and thus assumed to be fixed when analyzing the current labeled and unlabeled datasets.
However, our method and theory are agnostic to how $\hat Y_i$ is generated beyond the stated conditions.

While $\beta_0 (\tau)$ can be consistently estimated using only $\{ (Y_i, X_i) \}_{i\in \mathcal I_{\text{lab}}}$
via standard quantile regression \citep{Koenker2005},
we aim for more efficient estimation by leveraging the availability of $\hat Y_i$ for $i \in \mathcal I_{\text{unlab}}$.
To achieve this goal, we introduce the following two estimators.

\subsection{Score-debiasing convolution smoothed estimator (SD-CSE)}\label{sec:method:SD}
The score-debiasing (SD) estimator begins with the score equation for linear quantile regression in $\beta \in \mbb R^{p}$:
\begin{equation}
    \frac{1}{n} \sum_{i = 1}^n \psi (\beta; Y_i, X_i) = 0, \label{equation_unsmooth}
\end{equation}
where $\psi (\beta; y, x) := (\mbb I(y - x^{\top}\beta < 0) - \tau) x$.
Note that this equation follows from the population sub-gradient condition $\mbb E[(\mbb I(Y - X^{\top} \beta_0(\tau) < 0) - \tau) X ] = 0$.
To improve the inference from the limited labeled data, we incorporate $\hat Y_i$ by considering the augmented score equation \citep[see, e.g.,][]{MiaoMiaoWuZhaoLu2025jmlr}:
\begin{equation}\label{PPI_equation_unsmooth}
    \frac{1}{n} \sum_{i = 1}^n \psi (\beta; Y_i, X_i) - \hat W \left( \frac{1}{n} \sum_{i = 1}^n \psi(\beta; \hat Y_i, X_i) - \frac{1}{N} \sum_{i = n+1}^{n+N} \psi (\beta; \hat Y_i, X_i) \right) = 0,
\end{equation}
where the data-dependent weight matrix $\hat W$ controls the efficiency of the estimator.
To build intuition for the augmentation term,
observe that, when the predictions are perfect, $\hat Y_i = Y_i$ and thus $\psi (\beta; \hat Y_i, X_i) = \psi (\beta; Y_i, X_i)$, setting 
$\hat W$ to the identity reduces \eqref{PPI_equation_unsmooth} to $\frac{1}{N} \sum_{i = n+1}^{n+N} \psi (\beta; Y_i, X_i) = 0$.
Because the estimating equation now uses the larger sample size $N$, inference on $\beta_0$ becomes more efficient.
In contrast, even if $\hat Y_i$ is a poor predictor of $Y_i$, inference based on 
\eqref{PPI_equation_unsmooth} can be at least as efficient as \eqref{equation_unsmooth} by setting $\hat W$ to zero. 
It is also worth noting that the resulting estimator is asymptotically unbiased for any choice of $\hat{W}$, provided that $\hat{W}$ is bounded in probability, because two mean terms in the augmentation term asymptotically offset each other.
We will discuss how to choose $\hat W$ 
to optimize inference in Section~\ref{sec:theory:sd}.

From a computational standpoint, the discontinuity of the indicator function makes solving \eqref{PPI_equation_unsmooth} for $\beta$ quite challenging, especially when $p$ is not small.
Inspired by \cite{FernandesGuerreHorta2021jbes} and \cite{HePanTanZhou2023joe}, we sidestep this bottleneck by considering the following ``smoothed'' score:
\begin{equation}
    \hat \Psi_{h_n, \hat W} (\beta) 
    := \frac{1}{n} \sum_{i = 1}^n \psi_{h_n} (\beta; Y_i, X_i)
    - \hat W \left( 
    \frac{1}{n} \sum_{i = 1}^n \psi_{h_n} (\beta; \hat Y_i, X_i) -
    \frac{1}{N} \sum_{i = n+1}^{n+N} \psi_{h_n} (\beta; 
    \hat Y_i, X_i)
    \right),
\end{equation}
where $\psi_{h_n} (\beta; y, x) := (\mc K_{h_n} (x^{\top}\beta - y) - \tau) x$ and
$\mc K_{h_n} (s) := \mc K(s/{h}_n)$ smooths the indicator function in $\psi$. Specifically, $\mc K(s) := \int^s_{-\infty} K(t)dt$, where $K(\cdot)$ is a kernel and $h_n$ is a user-specified bandwidth that converges to zero as $n \rightarrow \infty$.
We choose the kernel $K (s)$ to be a probability density function so that $\mc K_{h_n} (s)$ converges pointwise to $\mbb I(0 < s)$ as $h_n \to 0$ for $s \neq 0$.
Now we define the score-debiasing convolution smoothed estimator(SD-CSE)
$\hat \beta_{\text{SD}, \hat W}$ as any root of the smoothed score equation:
\begin{equation}
    \hat \Psi_{h_n, \hat W} (\hat \beta_{\text{SD}, \hat W}) = 0. \label{SD_equation_smooth}
\end{equation}
By Lemma \ref{lem:root}, such a $\hat \beta_{\text{SD}, \hat W}$ exists with probability approaching one.
In practice, this continuous equation can be solved using the off-the-shelf software.
We also note that the standard quantile regression estimate based only on the labeled dataset
serves as a good initial value for solving the equation.

The choice of $\hat W$ is practically important because it determines the efficiency of the estimator.
In Section \ref{sec:theory:sd}, we show that $\sqrt{n} (\hat \beta_{\text{SD}, \hat W} -\beta_0(\tau))$ is asymptotically normal (Proposition \ref{prop:asy_normal_LD}),
with covariance depending on the 
probability 
limit of $\hat W$. To optimize efficiency gain, we also derive a closed-form expression for the optimal weight $W^*$ that minimizes this covariance (Proposition \ref{prop:opt_W}).
Accordingly, we recommend using a plug-in estimate of $W^*$ in practice.
For details on the estimation of $W^*$, see the discussion following Proposition \ref{prop:opt_W}.

Regarding the construction of confidence intervals, 
the plug-in estimator of the asymptotic variance can be unstable, as detailed in Section \ref{sec:theory:sd}, because it requires conditional density estimation.
We therefore develop a multiplier bootstrap procedure, inspired by \cite{PanZhou2021ii} and \cite{HePanTanZhou2023joe}. 
Namely, for each $b = 1, \dots, B$, where $B$ denotes the number of bootstrap replications, we carry out the following three steps.
First, we generate a random sample of non-negative weights $\{ w_{b, i} \}_{i = 1}^{n+N}$ satisfying $\mathbb E[w_{b, i}] = 1$
and $\Var (w_{b, i}) = 1$.
Examples of such weights include $1 + \xi$, where $\xi$ is a Rademacher random variable, and an exponential random variable with rate parameter one.
Second, we construct the bootstrapped score
\begin{equation}
    \hat \Psi_{h_n, \hat W}^{(b)} (\beta) 
    := \frac{1}{n} \sum_{i = 1}^n w_{b, i} \psi_{h_n} (\beta; Y_i, X_i)
    - \hat W \left( 
    \frac{1}{n} \sum_{i = 1}^n w_{b, i} \psi_{h_n} (\beta; \hat Y_i, X_i) -
    \frac{1}{N} \sum_{i = n+1}^{n+N} w_{b, i} \psi_{h_n} (\beta; 
    \hat Y_i, X_i)
    \right).
\end{equation}
Third, we obtain the bootstrapped estimator $\hat \beta_{\text{SD}, \hat W}^{(b)}$ as a root of the equation $\hat \Psi_{h_n, \hat W}^{(b)} (\beta) = 0$.
After $B$ replications,
we construct a bootstrap percentile confidence interval using the empirical percentiles of $\{ \hat \beta_{\text{SD}, \hat W}^{(b)} \}_{b = 1}^B$.

\subsection{Predict-then-debias convolution smoothed estimator (PTD-CSE)}\label{sec:method:ptd}
Another way to leverage the large amount of $\hat Y$ is the predict-then-debias (PTD) estimator.
Rather than correcting the estimating equation, this approach adds a correction term directly to the point estimate
obtained from the labeled data $\{(Y_i, X_i) \}_{i = 1}^n$. 
Namely, we first compute the convolution-smoothed quantile regression estimators
\citep[][]{FernandesGuerreHorta2021jbes,HePanTanZhou2023joe} as 
\begin{gather}
    \hat \beta_{\text{lab}} := \arg \min_{\beta \in \mathbb R^p} \sum_{i = 1}^n 
    \ell_{h_n} (\beta; Y_i, X_i), \quad
    \hat \gamma_{\text{lab}} := \arg \min_{\gamma \in \mathbb R^p}
    \sum_{i = 1}^n \ell_{h_n} (\gamma; \hat Y_i, X_i), \\
    \hat \gamma_{\text{unlab}} := \arg \min_{\gamma \in \mathbb R^p}
    \sum_{i = n+1}^{n+N} \ell_{h_n} (\gamma; \hat Y_i, X_i),
\end{gather}
where the convolution-smoothed check loss is
$\ell_{h_n} (b; y, x) := \int^{\infty}_{-\infty} \rho_{\tau} (v) K_{h_n} (v - y - x^{\top}\beta) dv$.
This loss is continuously differentiable and convex so that we 
can rely on the off-the-shelf software for the computation.
See \cite{HePanTanZhou2023joe} for detail.
We then construct the predict-then-debias convolution smoothed estimator (PTD-CSE) as 
\begin{equation}
    \hat \beta_{\text{PTD}, \tilde W} := \hat \beta_{\text{lab}} - \tilde W (\hat \gamma_{\text{lab}} - \hat \gamma_{\text{unlab}}),
\end{equation}
where the data-dependent weight matrix $\tilde W$ controls the efficiency, as in the SD-CSE.
A similar intuition applies to the PTD-CSE.
When prediction is perfect, $\hat Y_i = Y_i$ and $\hat \beta_{\text{lab}} = \hat \gamma_{\text{lab}}$,
so setting $\tilde W$ to the identity yields $\hat \beta_{\text{PTD}, \tilde W} = \hat \gamma_{\text{unlab}}$,
which uses the larger sample of size $N$ than $\hat \beta_{\text{lab}}$.
Meanwhile, we can always make inference at least as efficient as that based on $\hat \beta_{\text{lab}}$
by setting $\tilde W = 0$. 

Similar to the SD-CSE, $\sqrt{n} (\hat \beta_{\text{PTD}, \tilde W}-\beta_0(\tau))$ is asymptotically normal with a covariance matrix 
that depends on the limit of $\tilde W$ (Proposition \ref{prop:asy_normal_PTD}).
Hence, in practice, we recommend choosing $\tilde W$ to
optimize this asymptotic covariance matrix;
see Section \ref{sec:theory:ptd} for details.

For the same reason as for the SD-CSE, we implement a multiplier bootstrap procedure to construct confidence intervals for the PTD-CSE.
The implementation detail is analogous 
to that of the SD-CSE.
Namely, for each $b = 1, \dots, B$,
we carry out the following steps.
First, we generate a random sample of non-negative weights $\{ w_{b, i} \}_{i = 1}^{n+N}$ satisfying the same moment conditions as in the bootstrap for the SD-CSE.
We then compute the bootstrapped convolution-smoothed quantile regression estimators
\begin{gather}
\hat \beta^{(b)}_{\text{lab}} := \arg \min_{\beta \in \mathbb R^p} \sum_{i = 1}^n 
    w_{b, i} \ell_{h_n} (\beta; Y_i, X_i),\quad
\hat \gamma_{\text{lab}}^{(b)} := \arg \min_{\gamma \in \mathbb R^p}
    \sum_{i = 1}^n w_{b, i} \ell_{h_n} (\gamma; \hat Y_i, X_i), \\
\hat \gamma_{\text{unlab}}^{(b)} := \arg \min_{\gamma \in \mathbb R^p}
    \sum_{i = n+1}^{n+N} w_{b, i} \ell_{h_n} (\gamma; \hat Y_i, X_i),
\end{gather}
and obtain the bootstrapped PTD-CSE $\hat \beta_{\text{PTD}, \tilde W}^{(b)} := \hat \beta^{(b)}_{\text{lab}} - \tilde W (\hat \gamma_{\text{lab}}^{(b)} - \hat \gamma_{\text{unlab}}^{(b)}).$
Finally, after $B$ replications,
we obtain a bootstrap percentile confidence interval using the empirical percentiles of
$\{ \hat \beta_{\text{PTD}, \tilde W}^{(b)} \}_{b = 1}^B$.

As illustrated in the simulation study in Section~\ref{sec:simulation}, the check-loss-based PTD estimator exhibits overcoverage, whereas the PTD-CSE does not.
One possible explanation is that the nonsmoothness of the check-loss objective may adversely affect the finite-sample accuracy of the bootstrap approximation, a phenomenon related to earlier finding for $L_1$-type regression estimators; see, for example, \cite{AngelisHallYoung1993jasa}.
A detailed theoretical investigation of this issue is beyond the scope of this research.

\subsection{Ensemble}\label{sec:method:ensemble}
Choosing between the SD and PTD estimators is not easy.
Section 2.4 of \cite{KlugerLuZrnicWangBates2025arxiv} compares the asymptotic variances of the two estimators in the case $p=1$.
They conclude that neither the SD nor the PTD estimator is generally superior to the other, a finding that applies to the SD-CSE and PTD-CSE considered in the present study.
We therefore propose a heuristic approach based on a componentwise ensemble of these two estimators.
Let $\hat \beta_{\text{SD}, \hat W, (j)}$ and $\hat \beta_{\text{PTD}, \tilde W, (j)}$ denote the $j$-th components of 
$\hat \beta_{\text{SD}, \hat W}$
and $\hat \beta_{\text{PTD}, \tilde W}$,
respectively.
We then consider a linear combination
$c_j \hat \beta_{\text{SD}, \hat W, (j)} + (1 - c_j) \hat \beta_{\text{PTD}, \tilde W, (j)}$
for a constant $c_j$.
By a straightforward calculation, $c_j$ that minimizes the variance of the aforementioned combination is given as
\[
    c^*_j := \frac{\Var \left(\hat \beta_{\text{PTD}, \tilde W, (j)}\right)
    - \Cov \left(\hat \beta_{\text{SD}, \hat W, (j)}, \hat \beta_{\text{PTD}, \tilde W, (j)}\right)}{\Var \left( \hat \beta_{\text{SD}, \hat W, (j)} -
    \hat \beta_{\text{PTD}, \tilde W, (j)} \right)}.
\]
In practice, this $c_j^{\ast}$ is infeasible, so we estimate it via the bootstrap.
Namely, let $\{ (\hat \beta_{\text{SD}, \hat W}^{(b)}, \hat \beta_{\text{PTD},
\tilde W}^{(b)}) \}_{b = 1}^B$
be the bootstrap sample obtained as described in Sections \ref{sec:method:SD} and \ref{sec:method:ptd}.
Then we obtain the plug-in estimator $\hat c_j$ of $c_j^*$ by replacing the variances and covariance in $c_j^*$ with their bootstrap counterparts.
Consequently, we define the ensemble estimators as $\hat{\beta}_{\text{ENS}, (j)} := \hat c_j \hat \beta_{\text{SD}, \hat W, (j)} + (1 - \hat c_j) \hat \beta_{\text{PTD}, \tilde W, (j)}$
for $j = 1, \dots, p$.
This estimator is designed to approximate its oracle counterpart, $c_j^* \hat{\beta}_{\text{SD}, \hat{W}, (j)} + (1 - c_j^*)\hat{\beta}_{\text{PTD}, \tilde{W}, (j)}$, whose variance is no larger than those of $\hat{\beta}_{\text{SD}, \hat{W}, (j)}$ and $\hat{\beta}_{\text{PTD}, \tilde{W}, (j)}$, by its construction.
As discussed at the end of Section~\ref{sec:theory}, we do not establish the asymptotic properties of this ensemble estimator.
Nevertheless, it exhibits favorable finite-sample performance in the simulation study and data application in Sections~\ref{sec:simulation} and~\ref{sec:data_analysis}.

\section{Asymptotic properties}\label{sec:theory}

This section establishes the asymptotic properties of the SD-CSE and PTD-CSE.
For each estimator, we first discuss the required assumptions and then derive the asymptotic normality.
Throughout this section, we fix a quantile level $\tau \in (0, 1)$. Accordingly, we suppress the dependency of $\beta_0(\tau)$ on $\tau$ and just write $\beta_0$.
Let $K_{h_n}(s) := {h_n}^{-1}K(s/{h_n})$ denote the derivative of $\mc K_{h_n}(s)$. All limits below are taken with respect to both $n \rightarrow \infty$ and $N \rightarrow \infty,$ under the relative rate conditions stated below.

\subsection{SD-CSE: asymptotic theory}\label{sec:theory:sd}
To establish the asymptotic properties of the SD-CSE, we impose the following assumptions throughout this subsection.

\begin{assumption}\label{as:identification}
    The objective function $\beta \mapsto \mbb E[\rho_{\tau} (Y - X^{\top}\beta) - \rho_{\tau}(Y)]$
    has a unique minimizer $\beta_0$ over $\mbb R^p$.
\end{assumption}

\begin{assumption}\label{as:density}
    The conditional distribution of $Y$ given $X$ has a density function $f_{Y|X}$ with 
    respect to the Lebesgue measure that satisfies the following properties.
    $(a)$ There exists a finite constant $f_{Y, \sup}$ such that $\sup_{u \in \mbb R} f_{Y|X} (u) \leq f_{Y, \sup}$
    almost surely; 
    and
    $(b)$ there exists a finite constant $l_{Y}$ such that $|f_{Y|X} (u) - f_{Y|X}(v)| \leq l_{Y} |u - v|$ for all $u, v \in \mbb R$ almost surely.
\end{assumption}

\begin{assumption}\label{as:moment}
    $\mbb E[\| X \|^4]$ is finite. 
    Furthermore, $\mbb E[f_{Y | X} (X^{\top}\beta_0) X X^{\top}]$ is positive definite.
\end{assumption}

\begin{assumption}\label{as:kernel}
    $K(\cdot)$ is bounded and continuous, and satisfies:
    $(a)$ $K(u) = K(-u)$ and $K(u) \geq 0$ for all $u \in \mbb R$,
    $(b)$ $\int^{\infty}_{-\infty} K(u) du = 1$,
    $(c)$ $\kappa_j := \int^{\infty}_{-\infty} |u|^j K(u)du$ is finite for $j = 1, 2$, and
    $(d)$ $u^3 K(u)$ has finite limits as $|u| \to \infty$.
\end{assumption}

\begin{assumption}\label{as:h}
    $h_n \to 0$ and $n^{1/2} h_n \to \infty$; moreover, $h_n = o(n^{-1/4})$.
\end{assumption}

\begin{assumption}\label{as:y_hat}
    $\mathbb{P}(\hat{Y} \neq X^{\top}\beta_0) = 1$.
\end{assumption}

Assumption \ref{as:identification} is the standard identification condition for a possibly misspecified linear quantile regression model.
A closely related centered objective function is used in Equation~(3) of \cite{AngristChernozhukovFernandezval2006ecta}.
By Knight's identity (see, e.g., the proof of Theorem~1 of \cite{knight1998aos}),
\[
  \rho_{\tau}(Y-X^{\top}\beta) - \rho_{\tau}(Y) = -X^{\top}\beta (\tau - \mathbb{I}(Y < 0)) 
  + \int^{X^{\top}\beta}_0 (\mathbb{I}(Y < s) - \mathbb{I}(Y < 0)) ds.
\]
Thus, under Assumption~\ref{as:moment}, the expectation of this centered objective is finite without imposing $\mathbb{E}[|Y|] < \infty$.
The density and the moment conditions in Assumptions \ref{as:density} and \ref{as:moment} 
are stronger than the standard conditions used for check-loss-based quantile regression 
\citep[see e.g., Theorem 3 of][]{AngristChernozhukovFernandezval2006ecta}.
However, they are no more restrictive than those in existing works on convolution-smoothed quantile regression.
Specifically, Assumptions Q and A of \cite{FernandesGuerreHorta2021jbes} assume a smooth, positive conditional density and bounded covariates, 
while Theorem 4.2 of \cite{HePanTanZhou2023joe} requires a Lipschitz, bounded conditional density and sub-Gaussian design.
Assumption \ref{as:kernel} is standard in the convolution-smoothed quantile regression literature. 
Parts $(a)$--$(c)$ coincide with Condition 4.1 of \cite{HePanTanZhou2023joe} and are similar to the basic kernel requirements in Assumption K of \cite{FernandesGuerreHorta2021jbes}.
We add part $(d)$ to control tail behavior in the misspecified setting. 
Common choices such as Gaussian and uniform kernels satisfy these conditions. 
The bandwidth rate in Assumption \ref{as:h} 
is fast enough for the smoothing bias to vanish
asymptotically, while slow enough that the estimator 
converges faster than $h_n$;
it is comparable to the assumption in Theorem 4.3 of \cite{HePanTanZhou2023joe}.
Assumption \ref{as:y_hat} is introduced for the smoothing bias arising from approximating $\mathbb{I}(\hat{Y} - X^{\top}\beta_0 < 0)$ by $\mathcal{K}_{h_n} (X^{\top}\beta_0 - \hat{Y})$ in the augmentation term to vanish asymptotically in the proof of Proposition~\ref{prop:asy_normal_LD}.
Indeed, without this assumption, $\lim_{h_n \to 0} \mathcal{K}_{h_n}(X^{\top}\beta_0 - \hat{Y}) \neq \mathbb{I}(\hat{Y} - X^{\top}\beta_0 < 0)$ may hold with positive probability.
This assumption holds if the conditional distribution of $\hat Y$ given $X$ is continuous,
which is plausible when $\hat Y$ is produced by a pre-trained model $\hat f$ that uses auxiliary continuous features $Z$ in addition to $X$.
Even if the conditional distribution is not continuous, 
the assumption still holds except in the rare case where $X^{\top}\beta_0$ predicts $\hat Y$ 
perfectly with positive probability.

As a preliminary step toward the asymptotic normality, we first establish the following result on the consistency and the convergence speed of $\hat \beta_{\text{SD}, \hat W}$.
\begin{proposition}\label{prop:rate_PPI}
    Assume Assumption \ref{as:identification}, \ref{as:density}, \ref{as:moment}, \ref{as:kernel}, and \ref{as:h}. 
    If $\hat W = O_p(1)$ and $n/N = O(1)$, then
    $\hat \beta_{\text{SD}, \hat W} - \beta_0 = o_p(h_n)$.
\end{proposition}

The $o_p(h_n)$ rate is crucial for our asymptotic normality proof, as it makes the remainder term for the
asymptotic linear expansion negligible.
A similar rate result appears in Theorem 4.1 of \cite{HePanTanZhou2023joe};
however, our proof differs from theirs in two respects.
First, our estimator is not necessarily the minimizer of convex function,
a property \cite{HePanTanZhou2023joe} exploits in their proof.
Hence, we derive the convergence rate from the (smoothed) score equation, 
following a technique used in the threshold/change-plane regression literature \citep[e.g.,][]{SeoLinton2007joe,Takeishi2023sjs}.
Second, we allow misspecification of the linear quantile regression model, which 
makes the assessment of the smoothing bias in the population subgradient different from 
that in \cite{HePanTanZhou2023joe}.
Furthermore, unlike \cite{MiaoMiaoWuZhaoLu2025jmlr}, we do not impose a compactness assumption on 
the parameter space; this is made possible in part by the fact that
the population score in the limit can be written as the gradient of a convex function.

We now turn to the asymptotic distribution of $\sqrt{n} (\hat \beta_{\text{SD}, \hat W} - \beta_0)$.
First, we define the matrices that appear in the asymptotic covariance matrix.
\begin{equation}
    \begin{split}
        \Lambda_{\text{lab}} &:=  \mathrm{Var} ( \{ \mathbb I(Y - X^{\top}\beta_0 < 0) - \tau \} X) = \mathbb E[ \{ \mathbb I(Y - X^{\top}\beta_0 < 0) - \tau \}^2 X X^{\top}], \\
        \Lambda_{\text{unlab}} &:= \mathrm{Var} ( \{ \mathbb I(\hat Y - X^{\top}\beta_0 < 0) - \tau \} X ), \\
        \Lambda_{\text{cov}} &:= \mathrm{Cov} ( \{ \mathbb I(Y - X^{\top}\beta_0 < 0) - \tau \} X, \{ \mathbb I(\hat Y - X^{\top}\beta_0 < 0) - \tau \} X), \\
        H &:= \mathbb{E} [f_{Y|X}(X^{\top}\beta_0) X X^{\top}].
    \end{split}
\end{equation}
The following proposition establishes the asymptotic normality of $\sqrt{n} (\hat \beta_{\text{SD}, \hat W} - \beta_0)$.
\begin{proposition}\label{prop:asy_normal_LD}
    Assume Assumptions \ref{as:identification}, \ref{as:density}, \ref{as:moment}, \ref{as:kernel}, \ref{as:h}, and \ref{as:y_hat}. If $\hat W \to_p W_0$ for some matrix $W_0$ 
    and $n/N \to r$ for some finite constant $r$, then 
    \[ \sqrt{n} (\hat \beta_{\text{SD}, \hat W} - \beta_0) 
    \to_d 
    N(0, H^{-1} \Lambda(W_0) H^{-1}),\]
    where $\Lambda(W_0) := \Lambda_{\text{lab}} - \Lambda_{\text{cov}}W_0^{\top} 
    - W_0 \Lambda_{\text{cov}} + (1+r) W_0 \Lambda_{\text{unlab}} W_0^{\top}$.
\end{proposition}
The asymptotic covariance has the familiar sandwich form found in the literature  
\citep[e.g.,][]{MiaoMiaoWuZhaoLu2025jmlr}.
Similar to standard quantile regression, plug-in estimation of $H$
typically requires conditional density estimation, which could be unstable.
Hence, we recommend the bootstrap for inference.

From the form of $\Lambda (W_0)$ in Proposition \ref{prop:asy_normal_LD},
we can minimize the asymptotic covariance matrix with respect to $W_0$ in the sense of positive semi-definite matrix:
for two real matrices $W_1, W_2 \in \mathbb R^{p \times p}$, 
$W_1 \preceq W_2$ is defined to be that $W_2 - W_1$ is a positive semi-definite matrix.
\begin{proposition}\label{prop:opt_W}
    Assume the assumption of Proposition \ref{prop:asy_normal_LD},
    and that $\Lambda_{\text{unlab}}$ is invertible. Let 
    \begin{equation}
        W^{*} := (1+r)^{-1} \Lambda_{\text{cov}} \Lambda_{\text{unlab}}^{-1}.
    \end{equation}
    Then $H^{-1} \Lambda (W^{*}) H^{-1} \preceq H^{-1} \Lambda (W) H^{-1}$ for any $W \in \mathbb R^{p \times p}$.
\end{proposition}
From the forms of $\Lambda_{\text{cov}}$ and $\Lambda_{\text{unlab}}$, 
it is straightforward to construct a plug-in estimator of $W^*$.
Let $\hat \beta_{\text{init}}$ be a pilot consistent estimate of $\beta_0$
(e.g., the quantile regression estimate based on the labeled data only).
Then estimate $r$ by $n/N$ and replace $\Lambda_{\text{cov}}$ and $\Lambda_{\text{unlab}}$ by their sample analogs.

A caveat in estimating of $W^*$ is the invertibility of $\Lambda_{\text{unlab}}$. In some datasets, such as the one considered in Section~\ref{sec:data_analysis}, the predictor $\hat Y$ systematically over- or under-estimates $Y$, particularly in the lower or upper tails. In such cases, $\mathbb{I}(\hat Y - X^{\top}\beta_0(\tau) < 0)$ can be mostly if not all zero or one, which can make $\Lambda_{\text{unlab}}$ practically singular. 
To address this issue, we propose calibrating $\hat Y$ as follows.
Specifically, using the labeled data, we fit the linear quantile regression model $Y = \eta_1 + \eta_2 \hat Y$ at the target quantile level $\tau$, and let $\hat \eta = (\hat \eta_1, \hat \eta_2)$ denote the resulting estimator.
We then define the calibrated surrogate by $\hat Y(\hat \eta) := \hat \eta_1 + \hat \eta_2 \hat Y$, and use it as a new predictor of $Y$.
The calibrated surrogate $\hat Y (\hat \eta)$ is expected to provide a better match to $Y$ at the target quantile level.
In Appendix \ref{app:calibration}, we show that, under suitable regularity conditions, the calibration step is asymptotically negligible for the resulting SD-CSE estimator, in the sense that its asymptotic variance is the same as if the probability limit of $\hat\eta$ were known.

\subsection{PTD-CSE: asymptotic theory}\label{sec:theory:ptd}
This subsection derives the asymptotic distribution of the PTD-CSE.
In addition to Assumptions \ref{as:identification}, \ref{as:density}, \ref{as:moment}, \ref{as:kernel} and \ref{as:h} for the SD-CSE (Proposition \ref{prop:asy_normal_LD}),
we need the following regularity condition on $\hat Y$ (replacing Assumption \ref{as:y_hat}), which is used in the asymptotic expansion of $\hat \gamma_{\text{lab}}$
and $\hat \gamma_{\text{unlab}}$.
\begin{assumption}\label{as:identification_yhat}
    The objective function 
    $\gamma \mapsto \mathbb E[\rho_{\tau}(\hat Y - X^{\top}\gamma)-\rho_{\tau}(\hat{Y})]$
    has a unique minimizer $\gamma_0$ over $\mathbb R^p$.
\end{assumption}

\begin{assumption}\label{as:density_yhat}
    The conditional distribution of $\hat Y$ given $X$ has a density function $f_{\hat Y|X}$ with 
    respect to the Lebesgue measure that satisfies the following properties.
    $(a)$ There exists a finite constant $f_{\hat Y, \sup}$ such that $\sup_{u \in \mbb R} f_{\hat Y|X} (u) \leq f_{\hat Y, \sup}$
    almost surely; 
    and
    $(b)$ there exists a finite constant $l_{\hat Y}$ such that $|f_{\hat Y|X} (u) - f_{\hat Y|X}(v)| \leq l_{\hat Y} |u - v|$ for all $u, v \in \mbb R$ almost surely.
\end{assumption}

\begin{assumption}\label{as:moment_yhat}
    $\mbb E[f_{\hat Y | X} (X^{\top}\gamma_0) X X^{\top}]$ is positive definite.
\end{assumption}
These assumptions parallel Assumptions \ref{as:identification}, \ref{as:density} and \ref{as:moment}, but are imposed on the conditional distribution of $\hat Y$ given $X$.
They are stronger than Assumption \ref{as:y_hat} and are required because our proof relies on the asymptotic linear expansion of the smoothed quantile regression estimators for $\hat Y$ on $X$.

To establish the asymptotic normality of $\sqrt{n} (\hat \beta_{\text{PTD}, \tilde W} - \beta_0)$, we first define the matrices that appear in the asymptotic covariance matrix.
\begin{equation}
   \begin{split}
    \Sigma_{\text{lab}} = \ &H^{-1} 
    \mbb E[  \{\mbb I (Y - X^{\top}\beta_0 < 0) -  \tau\}^2 X X^{\top}] 
    H^{-1}, \\
    \Sigma_{\text{unlab}} = \ &J^{-1} 
    \mbb E[\{ \mbb I(\hat Y - X^{\top}\gamma_0 < 0) - \tau \}^2 X X^{\top} ] 
    J^{-1}, \\
    \Sigma_{\text{cov}} = \
    &H^{-1} 
    \mbb E[ \{ \tau - \mbb I(Y - X^{\top}\beta_0 < 0)\}\{\tau - \mbb I (\hat Y - X^{\top}\gamma_0 < 0)\} X X^{\top}] J^{-1},
    \end{split}
\end{equation}
where $J$ = $\mathbb E[f_{\hat Y | X} (X^{\top} \gamma_0) X X^{\top}]$, which is the analog of $H$ with the conditional density of the predicted outcome.

Then the following proposition gives the asymptotic normality.
\begin{proposition}\label{prop:asy_normal_PTD}
    Assume Assumptions \ref{as:identification}, \ref{as:density}, \ref{as:moment}, \ref{as:kernel}, \ref{as:h}, \ref{as:identification_yhat}, \ref{as:density_yhat},
    and \ref{as:moment_yhat}.
    If $\tilde W \to_p W_0$ 
    for some matrix $W_0$ and 
    $n/N \to r$ for some finite constant $r$, then
    \[
        \sqrt{n} (\hat \beta_{\text{PTD}, \tilde W} - \beta_0) \to_d 
        N(0, \Sigma_{\text{lab}} - \Sigma_{\text{cov}} W_0^{\top} - W_0 \Sigma_{\text{cov}}^{\top}
        + (1+r) W_0 \Sigma_{\text{unlab}} W_0^{\top}).
    \]
\end{proposition}
The proof follows a similar argument to Proposition 2.1 of \cite{KlugerLuZrnicWangBates2025arxiv}, which is 
based on the asymptotic linear expansions for $\sqrt{n} (\hat \beta_{\text{lab}} - \beta_0)$,
$\sqrt{n} (\hat \gamma_{\text{lab}} - \gamma_0)$, and $\sqrt{N} (\hat \gamma_{\text{unlab}} - \gamma_{0})$.
Since $\hat \beta_{\text{lab}}$, $\hat \gamma_{\text{lab}}$ and $\hat \gamma_{\text{unlab}}$ each 
solve a smoothed estimating equation analogous to \eqref{SD_equation_smooth}, 
the proof for the expansions largely follows that of Proposition \ref{prop:rate_PPI} and \ref{prop:asy_normal_LD}.
The plug-in covariance estimate relies on conditional density estimation and can be unstable, so
we recommend bootstrapping for inference similarly, as with the SD estimator.

Similar to the SD estimator, the asymptotic covariance matrix in 
Proposition \ref{prop:asy_normal_PTD} can be minimized with respect to $W_0$
in the sense of positive semi-definite matrix.
By section 2.3 of \cite{KlugerLuZrnicWangBates2025arxiv}, 
the optimal matrix is $(1 + r)^{-1} \Sigma_{\text{cov}} \Sigma_{\text{unlab}}^{-1}$.
Although a plug‑in estimator is available, 
it requires conditional density estimation; 
a more practical approach is to rely on the bootstrap.
Note that $\Sigma_{\text{cov}}$ is the asymptotic covariance between $\sqrt{n}(\hat \beta_{\text{lab}} - \beta_0)$ and $\sqrt{n} (\hat \gamma_{\text{lab}} - \gamma_0)$, while $\Sigma_{\text{unlab}}$ is the asymptotic variance of $\sqrt{N} (\hat \gamma_{\text{unlab}} - \gamma_0)$. 
Hence, inspired by Section 3.3 of \cite{KlugerLuZrnicWangBates2025arxiv}, 
we estimate $\Sigma_{\text{cov}}$ using 
the covariance of the bootstrapped counterparts of $\hat \beta_{\text{lab}}$,
and $\hat \gamma_{\text{lab}}$, and 
estimate $\Sigma_{\text{unlab}}$ 
using the variance of the bootstrapped counterparts of
$\hat \gamma_{\text{unlab}}$.
Specifically, these estimates are computed from the multiplier bootstrap sample $\{ (\hat \beta^{(b)}_{\text{lab}}, \hat \gamma^{(b)}_{\text{lab}},
\hat \gamma^{(b)}_{\text{unlab}}) \}_{b = 1}^B$ obtained as described in Section \ref{sec:method:ptd}.

Finally, we do not pursue an asymptotic analysis of the ensemble estimator introduced in Section~\ref{sec:method:ensemble}. Such an analysis would require additional technical arguments, including consistency of the bootstrap estimator of the oracle weight $c_j^*$. We leave this question for future research.

\section{Monte carlo simulation}\label{sec:simulation}
This section examines the finite-sample performance of the SD-CSE and PTD-CSE.
Specifically, we assess the empirical coverage and length of the resulting confidence intervals.
Throughout this section, we set $\mc K(\cdot)$ to the c.d.f. of the standard normal distribution and 
set $h_n$ to $\{ (p + \log n)/n \}^{2/5}$, as suggested in section 5 of \cite{HePanTanZhou2023joe}.

The data generating process in the simulation is partly based on section 5 of \cite{HePanTanZhou2023joe}.
We generate the covariates $X_i = (X_{i, 1}, X_{i, 2}, X_{i, 3})^{\top}$ from
a multivariate uniform distribution on $[-\sqrt{3}, \sqrt{3}]^3$
with covariance matrix $(0.7^{|j-k|})_{1 \leq j, k \leq 3}$,
using the R package \texttt{MultiRNG} \citep{falk1999cssc}.
We also generate variables $Z_i = (Z_{i, 1}, Z_{i, 2}, Z_{i, 3})^{\top}$, used in the prediction model $\hat Y = \hat f(X_i, Z_i)$,
as mutually independent mean-zero normal random variable, independent of $X_i$, with variance $\sigma_Z^2$.
The random noise $\varepsilon_i$ for quantile regression is decomposed as $\varepsilon_i = Z_{i, 1} + Z_{i, 2} + Z_{i, 3} + u_i$,
where $u_i$ follows mean-zero normal distribution with variance $\sigma^2_u$ and is independent of $X_i$ and $Z_i$.
We consider three choices of $(\sigma_{Z}^2, \sigma_{u}^2)$:
$(1, 1), (2/3, 2),$ and $(1/3, 3)$, holding $\Var(\varepsilon_i) = 4$ fixed, which helps examine how the predictive power of $\hat Y$ affects the results.
Given $\tau \in \{0.5, 0.75, 0.9\}$, we then generate the outcome variable $Y_i$ from the heteroskedastic model:
\begin{equation}
    Y_i = 1 + X_{i, 1} + X_{i, 2} + X_{i, 3} + (0.5 X_{i, 1} + 1) (\varepsilon_i - F^{-1}_{\varepsilon_i} (\tau)), \label{heteroskedastic}
\end{equation}
where $F^{-1}_{\varepsilon_i} (\cdot)$ is the quantile function of $\varepsilon_i$. Because $0.5 X_{i, 1} + 1 > 0$, the model \eqref{heteroskedastic} satisfies the linear quantile regression specification
$Q_{\tau}(Y_i|X_i) = 1 + X_{i, 1} + X_{i, 2} + X_{i, 3}$.
To train the prediction model $\hat f (\cdot)$, we generate a
single training sample $\{ Y_i, X_i, Z_i \}_{i = 1}^{n_{\text{train}}}$
of size $n_{\text{train}} = 10000$ according to the above data generating process.
Using this sample, we fit $\hat f(\cdot)$ via a neural network to estimate the conditional mean function $\mathbb E[Y_i|X_i, Z_i]$, implemented with the R package \texttt{nnet} \citep{VenablesRipley2002}.
We then generate 1000 Monte Carlo replications consisting of a labeled dataset $\{ Y_i, \hat Y_i, X_i \}_{i = 1}^n$ and an unlabeled dataset 
$\{ \hat Y_i, X_i \}_{i = n+1}^{10000}$, where $\hat Y_i = \hat f(X_i, Z_i)$ and $n \in \{250, 500, 1000\}$.

For the SD-CSE,
we consider two variants that 
differ in the choice of $\hat W$: (i) ``SD'', with $\hat W = I$ (the identity matrix), and 
(ii) ``SD-OPT'', with $\hat W$ given by an estimate of the optimal matrix.
We define ``PTD'' and ``PTD-OPT'' for the PTD-CSE in the same manner.
We also implement ``ENS,''
the componentwise ensemble of ``SD-OPT'' and ``PTD-OPT'' defined in Section~\ref{sec:method:ensemble}.
To solve the nonlinear equation for the SD-CSE, we use the R package \texttt{nleqslv}.
For the intermediate estimates $\hat{\beta}_{\text{lab}}$, $\hat{\gamma}_{\text{lab}}$, 
and $\hat{\gamma}_{\text{unlab}}$ used in
the PTD-CSE, we employ the R package \texttt{conquer}.
We compare the proposed methods against the following three benchmarks.
\begin{enumerate}
    \item ``LAB'':
convolution-smoothed quantile regression applied only to the labeled sample \citep{HePanTanZhou2023joe}
    \item ``PTD-CHECK'': the PTD estimator based on the check loss with the identity weight matrix
    \item ``PTD-CHECK-OPT'': the PTD estimator based on the check loss with an estimated optimal weight matrix.
\end{enumerate}
Note that the latter two are implemented in \cite{KlugerLuZrnicWangBates2025arxiv}.
Recall that the SD estimator based on the check loss is computationally costly and unreliable, as it requires solving a discontinuous equation. Therefore, we do not implement this approach here.
For all the methods, we construct 
bootstrap percentile
confidence intervals using the multiplier bootstrap with Rademacher weights
\citep{PanZhou2021ii, HePanTanZhou2023joe}.
The number of bootstrap replications is set to 1000 throughout.

We first examine the empirical coverage of the 95\% confidence intervals.
Figure~\ref{fig:sim-coverage} reports, for each method, empirical coverage averaged over the three slope coefficients ($X_1$, $X_2$, and $X_3$; intercept excluded) across different choices of $\tau$, $n$, and $(\sigma_Z^2,\sigma_u^2)$.
While ``SD-OPT'' and ``ENS'' exhibit empirical coverage close to .95, the other proposed methods,``SD,'' ``PTD,'' and ``PTD-OPT,'' are generally conservative, with slightly higher coverage rate, but much less so than ``PTD-CHECK'' and ``PTD-CHECK-OPT.''
The overcoverage for ``PTD-CHECK'' and ``PTD-CHECK-OPT'' is consistent with the pattern observed in \cite{KlugerLuZrnicWangBates2025arxiv}.

\begin{figure}[t]
    \centering
    \includegraphics[width=0.95\textwidth]{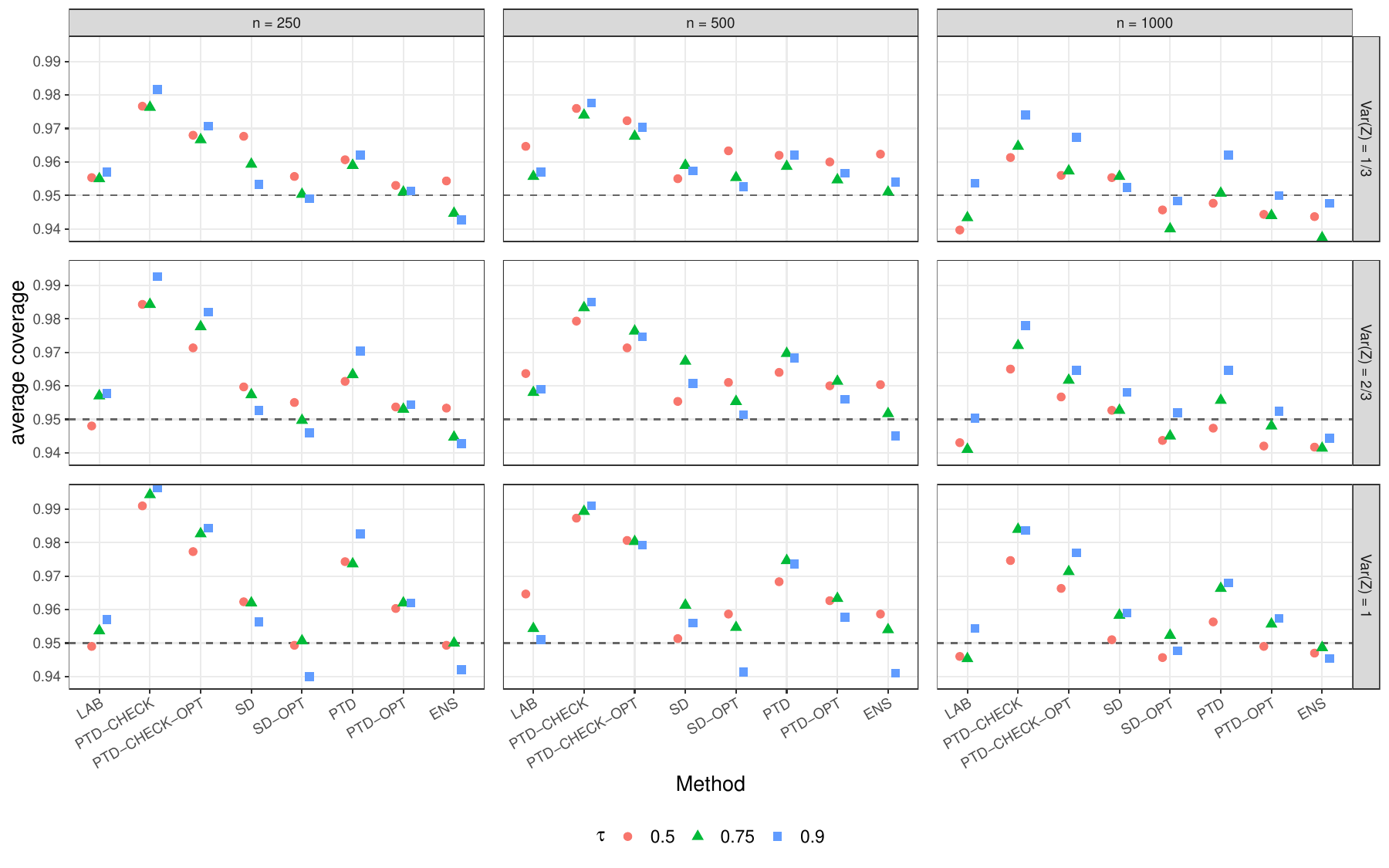}
    \caption{Empirical coverage for 95\% confidence intervals across methods. Each point shows the coverage averaged over the three slope coefficients ($X_1$, $X_2$, and $X_3$; intercept excluded), with color indicating $\tau \in \{0.5, 0.75, 0.9\}$, across sample sizes and variance levels.}
    \label{fig:sim-coverage}
\end{figure}

We now turn to the length of the confidence intervals.
Figure \ref{fig:sim-length} illustrates how interval length, averaged across the coefficients of $X_1$, $X_2$, and $X_3$, varies across methods as we change (i) the labeled sample size and (ii) the noise levels, for $\tau = 0.5, 0.75$, and $0.9$, respectively.
Although empirical coverage rates are not identical across methods, comparing interval lengths remains practically relevant because these are the intervals produced in practice.
Under this averaged metric, the the reductions in interval length for ``SD-OPT'' and ``PTD-OPT'' relative to ``LAB'' are more pronounced when $\sigma_{Z}^{2}$ is larger.
This is intuitive because, with $\Var(\varepsilon)$ fixed, increasing $\sigma_Z^{2}$ strengthens the signal in $Z_i$, making $\hat Y_i=\hat f(X_i,Z_i)$ better able to explain the residual variation in $Y_i$ beyond $X_i$.
Furthermore, even in the lowest-$\sigma_{Z}^2$ setting ($\sigma_Z^2 = 1/3$ and $\tau = 0.9$), ``SD-OPT'' performs at least as well as ``LAB'', and ``PTD-OPT'' slightly outperforms ``LAB''.
Also consistent with the overcoverage observed above, the confidence intervals from ``PTD-CHECK'' and ``PTD-CHECK-OPT'' are noticeably wider than those from ``SD-OPT'' and ``PTD-OPT''.
Note that ``ENS'' consistently outperforms both “SD-OPT” and “PTD-OPT”. 
\begin{figure}[t]
    \centering
    \includegraphics[width=0.95\textwidth,height=0.78\textheight,keepaspectratio]{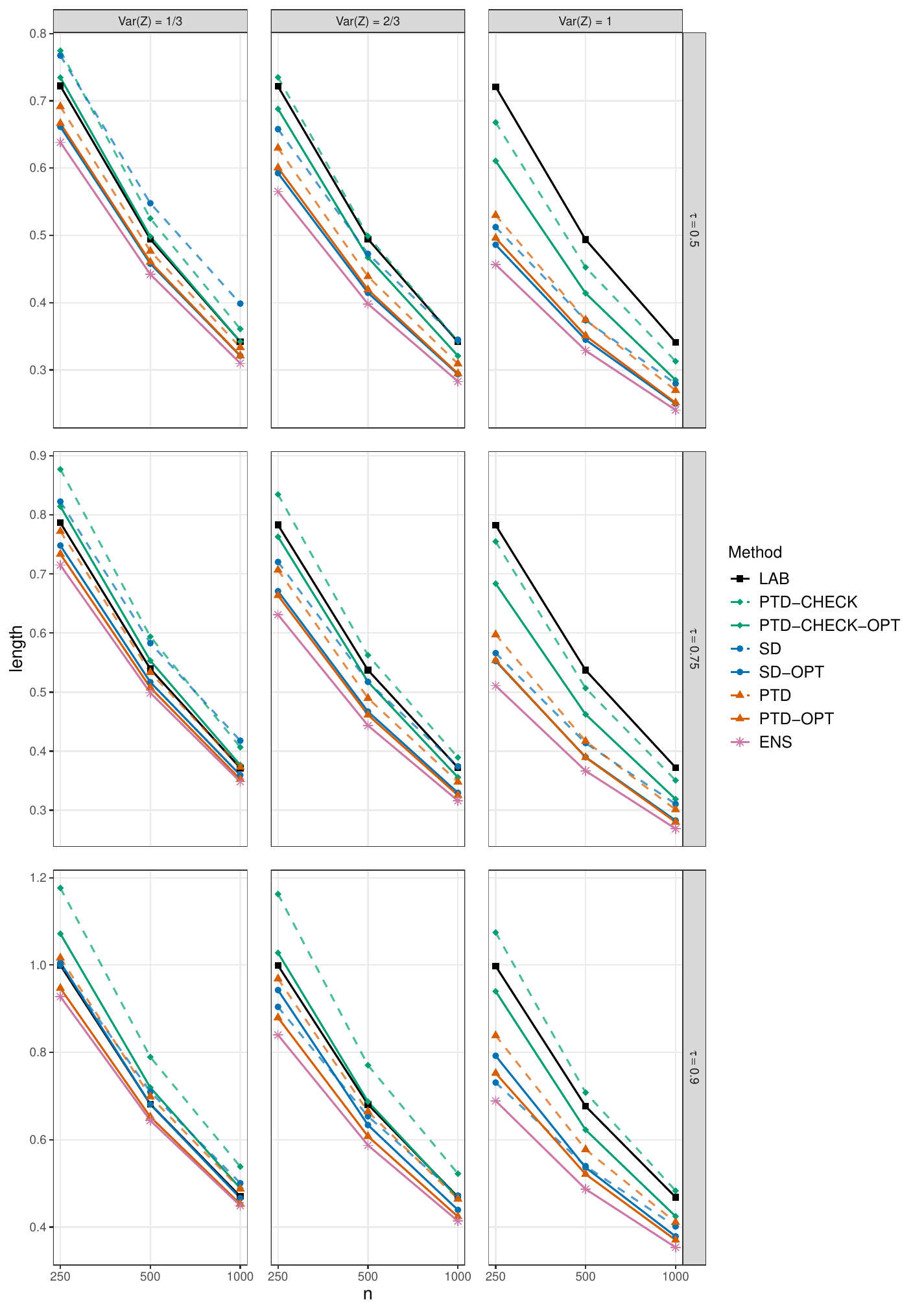}
    \caption{Confidence interval length averaged across the coefficients $X_1$, $X_2$, and $X_3$, over 1000 replications.}
    \label{fig:sim-length}
\end{figure}

While the averaged coverage and length summaries clarify overall efficiency patterns, they can obscure coefficient-specific behavior. In particular, at $\tau = 0.9$, the coefficient on $X_1$, which drives heteroskedasticity in \eqref{heteroskedastic}, shows noticeable undercoverage for ``SD-OPT'' and, consequently, for ``ENS,'' at smaller sample sizes. This undercoverage diminishes as the degree of heteroskedasticity decreases or the sample size increases. See Appendix~\ref{app:undercoverage} for detail.

We also conduct a brief sensitivity analysis of the SD-CSE and PTD-CSE with respect to the choice of $h_n$.
Overall, the performance of these estimators is not sensitive to $h_n$, which is consistent with the robustness of convolution-smoothed quantile regression observed in the numerical studies of \cite{HePanTanZhou2023joe}. We defer the detail to Appendix \ref{app:sensitivity}.
Finally, we assess the performance of ``SD-OPT'' in a setting where estimating $W^*$ is difficult and the calibration procedure introduced in Section \ref{sec:method:SD} is applied. The results are generally similar to those obtained above in the setting without calibration, both in terms of empirical coverage probabilities and interval lengths. See Appendix \ref{app:calibration} for details.

\section{Application to the St.~Louis housing price dataset}\label{sec:data_analysis}

In this section, we illustrate the proposed method using housing price data  recorded during 2025 from selected municipalities in the St.~Louis area in the United States.
In particular, we examine how the sale price of single-family homes is associated with the assigned public school district, across different quantiles, after adjusting for other covariates.
Note that sale prices are observed only when a property is sold; for the majority of off-market houses, actual sale prices are unavailable.
Therefore, we use the appraisal value of these single-family homes as a predictor, obtained from St.~Louis County tax records of the previous year.
As of the end of 2024, there were 21,007 single-family homes in the selected area  after excluding properties with outlying house sizes; of these, 775 were sold in 2025 and are treated as labeled data, while the remaining 20,232 off-market homes are treated as unlabeled data. The appraisal value is observed for all properties, along with a set of covariates such as school district, house size, and building age.

We consider the following linear quantile regression specification:
\begin{equation}
    \begin{split}
        \text{price} =
        \beta_1(\tau)
        + \beta_2(\tau)\,\text{school}
        + \beta_3 (\tau)\, \text{school} \times \text{size}
        + \beta_4(\tau)\, \text{size}
        + \beta_5(\tau)\,\text{age}
        + \beta_6(\tau)\,\text{age}^2
        + \varepsilon_{\tau},
    \end{split}
\end{equation}
for $\tau \in \{0.1, 0.25, 0.5, 0.75, 0.9\}$,
where ``price'' denotes the house price in thousands of U.S. dollars, 
``size'' denotes the house size measured in square feet,
and ``age'' denotes the building age of the house.
Both ``size'' and ``age'' are standardized to have mean zero and unit variance.
The variable ``school''  is a binary indicator for whether a property is located in a relatively high-performing school district.
Specifically, it is defined to be 1 if the property is located in Clayton or Ladue school district, and 0 if it is located in University City.
This classification is based on publicly available district performance information from the Missouri Department of Elementary and Secondary Education (DESE) \citep{missouri_dese_apr}.
We include the quadratic term ``$\text{age}^2$'' to allow for a nonlinear effect of building age on price.
An interaction term between ``school'' and ``size'' is included to allow the school premium to vary with house size.
Table~\ref{tab:summary_df} reports summary statistics for the variables used in the analysis.

\begin{table}[htbp]
\centering
\caption{Summary statistics for variables}
\label{tab:summary_df}
\begin{tabular}{lrrrrr}
\hline
Variable & Mean & Std.\ Dev. & Min & Median & Max \\
\hline
Price (\$1,000) & 735.6 & 608.1 & 37.0 & 577.0 & 4,800.0 \\
Appraised value (\$1,000) & 543.9 & 406.7 & 11.6 & 457.2 & 10,925.9 \\
Age & 74.63 & 24.20 & 0 & 75 & 217 \\
Size (sq ft) & 2,299.1 & 1,186.1 & 528 & 2,018 & 6,000 \\
School score high (0/1) & 0.521 & 0.500 & 0 & 1 & 1 \\
\hline
\end{tabular}

\begin{flushleft}
\footnotesize Notes: Price is based on labeled data ($n=775$). Other variables are based on all the observations of size $n + N =21{,}007$.
\end{flushleft}
\end{table}

In this dataset, $\hat Y$ (the appraised value from the past year) tends to underestimate $Y$ (the house price), particularly in the upper half of the distribution; see Figure~\ref{fig:price_vs_appraised} for a scatterplot of these two variables.
This can lead to numerical instability in estimating the optimal weight matrix for the SD-CSE when targeting higher quantiles, as discussed after Proposition~\ref{prop:opt_W}.
To address this potential issue, we apply the calibration of $\hat Y$ proposed there for $\tau \in \{ 0.5, 0.75, 0.9 \}$.
\begin{figure}[htbp]
    \centering 
    \includegraphics[width=0.7\textwidth]{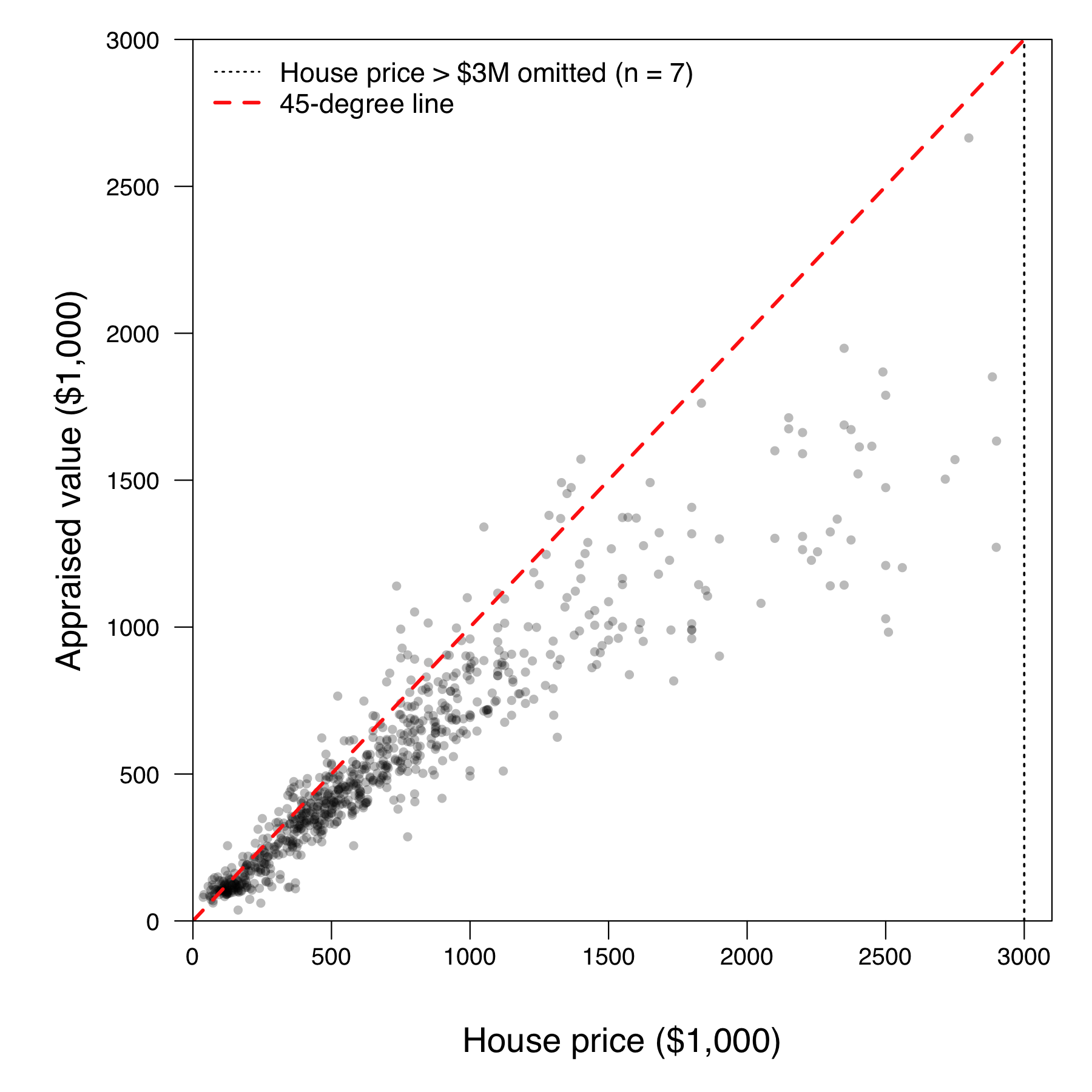}
    \caption{House price versus appraised value.}
    \label{fig:price_vs_appraised}
\end{figure}

We compare the proposed methods---the SD-CSE with the estimated optimal weight matrix (``SD-OPT''), the PTD-CSE with the estimated optimal weight matrix (``PTD-OPT''), and their ensemble (``ENS'')---with convolution-smoothed quantile regression applied only to the labeled sample (``LAB'').
The implementation details are the same as those in Section \ref{sec:simulation}, except for the calibration described in the previous paragraph.
Figure \ref{fig:school_and_school_size_ci} reports the confidence intervals and point estimates for the coefficients on ``school'' and its interaction with ``size,'' respectively, while Table \ref{tab:ci_lengths_combined} reports the corresponding interval lengths.
According to Table~\ref{tab:ci_lengths_combined}, both ``SD-OPT'' and ``PTD-OPT'' consistently outperform ``LAB.''
``ENS'' performs close to the better of ``SD-OPT'' and ``PTD-OPT,'' outperforming ``LAB'' and providing a safeguard against the worse of the two. In terms of interpretation, because ``size'' is standardized, the coefficient on ``school'' represents the estimated school-district premium for an average-size house, holding the other covariates fixed.
Figure~\ref{fig:school_and_school_size_ci} suggests that this premium is positive and larger for more expensive homes, while the positive upper-quantile coefficients on ``school $\times$ size'' indicate a larger premium among 
larger homes.

\begin{figure}[htbp]
\centering
\includegraphics[width=\textwidth]{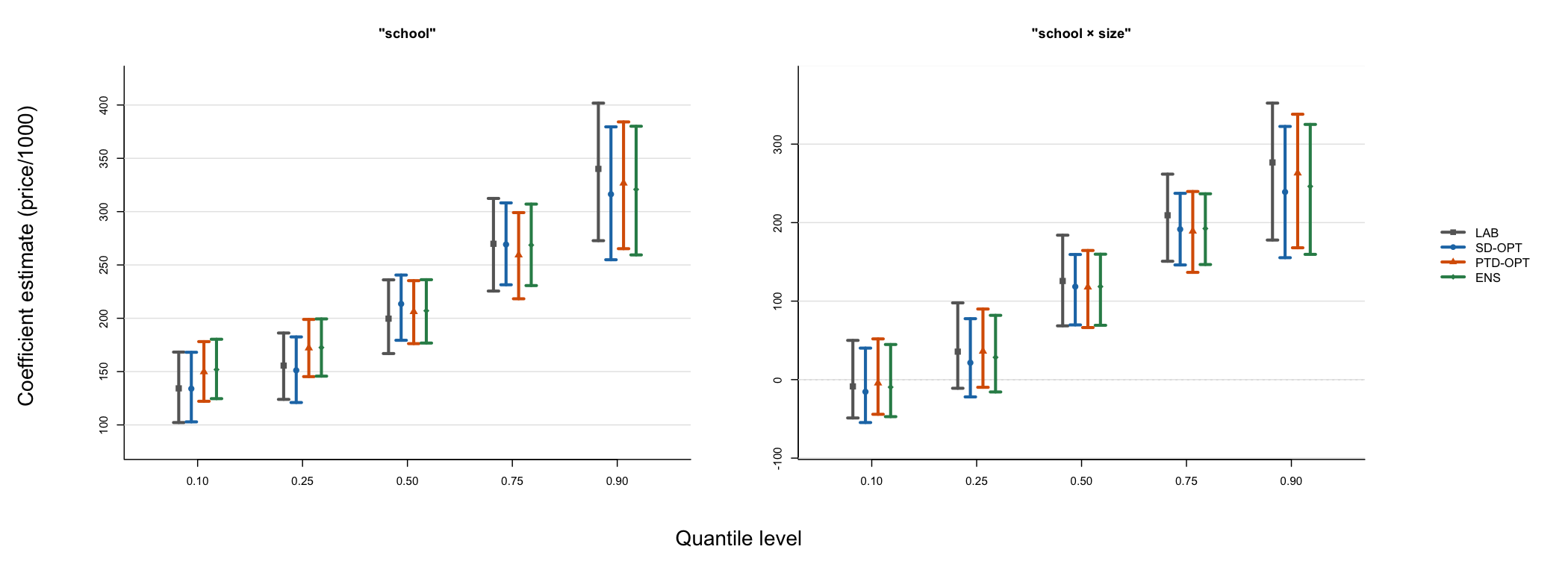}
\caption{Point estimates and 95\% confidence intervals across quantile levels for the coefficients on \texttt{school} and \texttt{school \(\times\) size}.}
\label{fig:school_and_school_size_ci}
\end{figure}

\begin{table}[htbp]
\centering
\caption{Lengths of 95\% confidence intervals (price/1000)}
\label{tab:ci_lengths_combined}
\begin{tabular}{llccccc}
\hline
Coefficient & Method & \multicolumn{5}{c}{Quantile level ($\tau$)} \\
\cline{3-7}
 &  & 0.10 & 0.25 & 0.50 & 0.75 & 0.90 \\
\hline
school & LAB     & 66.08 & 62.24 & 69.12 & 86.86 & 129.00 \\
       & SD-OPT  & 65.26 & 61.53 & 61.19 & 76.80 & 124.63 \\
       & PTD-OPT & 55.96 & 53.66 & 59.06 & 80.81 & 118.90 \\
       & ENS     & 55.71 & 53.69 & 59.41 & 76.40 & 120.66 \\
\hline
school $\times$ size & LAB     & 98.84 & 108.78 & 115.49 & 110.97 & 174.45 \\
                     & SD-OPT  & 94.82 & 99.64  & 89.72  & 91.24  & 167.20 \\
                     & PTD-OPT & 96.06 & 99.72  & 98.19  & 103.01 & 170.12 \\
                     & ENS     & 91.91 & 97.61  & 90.58  & 90.13  & 165.52 \\
\hline
\end{tabular}

\medskip
    \parbox{0.9\linewidth}{\footnotesize Note: Each entry reports the length of 95\% confidence interval for the corresponding method, quantile and coefficient.}
    
\end{table}

\section{Conclusion}\label{sec:conclusiont}
In this paper, we study quantile regression in settings where the gold-standard outcome is observed only for a small labeled sample, while surrogate variables for that outcome are abundant.
A straightforward application of the existing ``prediction-powered inference'' method has two drawbacks: computational difficulty and potential overcoverage of the resulting confidence intervals.
Our approach addresses these issues by smoothing the check loss in quantile regression, which leads to more tractable score functions and mitigates overcoverage.
We establish the asymptotic distributions of the resulting estimators under possible model misspecification and further propose a heuristic ensemble that combines them.
Through simulation studies and a data application, we find that the proposed methods generally shorten confidence intervals, while maintaining reliable coverage in the setting considered.

Several interesting directions remain for future research. First, although the calibration step used to stabilize estimation of the optimal weight matrix in SD-CSE appears to work reasonably well, it is rather heuristic. It would be useful to develop a more systematic approach to post-processing surrogate predictions, not only to address matrix-inversion issues but also to improve the overall performance of the prediction-powered inference  more broadly. Second, our theory assumes that the labeled and unlabeled samples are drawn from the same population. In many applications, however, labeled data may be collected under a different sampling mechanism, leading to distributional shifts between the labeled and unlabeled samples. Developing valid and efficient inferential procedures for such settings is therefore an important practical direction.

\clearpage
\appendix
\section*{Appendix}
\addcontentsline{toc}{section}{Appendix}
Throughout the appendix, we use the following notation.
Let $\|\cdot\|$ denote the Euclidean norm.
For any real sequences $a_n$ and $b_n$, we write $a_n \lesssim b_n$ if $a_n \leq \mathcal{C} b_n$ for some constant $\mathcal{C}$ independent of $n$.
For any random variable $U$ and random vector $V$, let $f_{U \mid V}$ and $F_{U \mid V}$ denote the conditional density function and conditional distribution function of $U$ given $V$, respectively.
All limits are taken as $n, N \to \infty$.
We also use the following empirical process notation. 
For any random variable $Z$ and
a measurable function $g$, $\mbb P_n$, $\mbb P_N$, and $P$ are operators such that
\[
    \mbb P_n g(Z) = \frac{1}{n} \sum_{i = 1}^n g(Z_i), \quad
    \mbb P_N g(Z) = \frac{1}{N} \sum_{i = n+1}^{n+N} g(Z_i), \quad
    Pg(Z) = \mbb E[g(Z)].
\]
Furthermore, operators $\mbb G_n$ and $\mbb G_N$ are defined as 
\begin{equation}
    \mbb G_n g(Z) = \sqrt{n} (\mbb P_n - P) g(Z), \quad \mbb G_N g(Z) = \sqrt{N} (\mbb P_N - P) g(Z).
\end{equation}

\mathtoolsset{showonlyrefs=true}
\numberwithin{equation}{section}

\section{Proofs}\label{app:proofs}

\begin{proof}[Proof of Proposition \ref{prop:rate_PPI}]
    The proof is divided into two steps. 
    First we show the consistency of $\hat \beta_{\text{SD}, \hat W}$: $\hat \beta_{\text{SD}, \hat W} \to_p \beta_0$.
    Then, we establish the convergence rate built on this consistency result.

    \underline{$(i)$ consistency}
    By Theorem 5.9 of \cite{vaart1998}, it suffices to check
    $(a) \sup_{\beta \in \mbb R^p} \| \hat \Psi_{h_n, \hat W} (\beta) - \mbb E[X(\mbb I(Y - X^{\top}\beta < 0) - \tau)] \| \to_p 0$,
    and $(b) \inf_{\beta: \| \beta - \beta_0 \| \geq \ve} \| \mbb E[X (\mbb I(Y - X^{\top}\beta < 0) - \tau)] \| > 0 = \| \mbb E[X (\mbb I(Y - X^{\top}\beta_0 < 0) - \tau)] \|$
    for any $\ve > 0$. Part $(a)$ is shown in Lemma \ref{lem:uniform consistency}.
    So it remains to prove part $(b)$. 
    Observe that $\| \mbb E[X(\mbb I(Y - X^{\top}\beta_0 < 0) - \tau)] \| = 0$ by Assumption \ref{as:identification} and Lemma \ref{lem:grad_beta}.
    Hence, we only need to show the inequality part.
    By Lemma \ref{lem:grad_beta}, 
    $\nabla_{\beta} \mbb E[\rho_{\tau} (Y - X^{\top}\beta) - \rho_{\tau}(Y)] = \mbb E[X(\mbb I(Y - X^{\top}\beta < 0) - \tau)]$.
    Furthermore, by the proof of Lemma \ref{lem:root}, $\nabla_{\beta} \mbb E[X(\mbb I(Y - X^{\top}\beta < 0) - \tau)] = \mbb E[X X^{\top} f_{Y|X} (X^{\top}\beta)]$,
    which is positive semi-definite.
    Hence, applying Lemma \ref{lem:convex} with $M(\beta) = \mbb E [\rho_{\tau} (Y - X^{\top}\beta)-\rho_{\tau}(Y)]$
    and $D(\beta) = \mbb E[X(\mbb I(Y - X^{\top}\beta < 0) - \tau)]$ in conjunction with Assumption \ref{as:identification},
    the desired inequality part is shown. Therefore, $\hat \beta_{\text{SD}, \hat W} \to_p \beta_0$ follows.

    \underline{$(ii)$ convergence rate} 
    By Lemma \ref{lem:root}, $\hat \Psi_{h_n, \hat W} (\hat \beta_{\text{SD}, \hat W}) = 0$
    with probability approaching one. 
    Consequently, $A_1 + A_2 + A_3 = 0$ holds with probability approaching one, where 
    \begin{equation}
        \begin{split}
            A_1 &:= \hat \Psi_{h_n, \hat W} (\hat \beta_{\text{SD}, \hat W}) - \mbb E[X (\mc K_{h_n} (X^{\top} \beta - Y) - \tau)]\mid_{\beta = \hat \beta_{\text{SD}, \hat W}}, \\
            A_2 &:= \mbb E[X(\mc K_{h_n}(X^{\top}\beta - Y) - \mc K_{h_n} (X^{\top}\beta_0 - Y))]\mid_{\beta = \hat \beta_{\text{SD}, \hat W}}, \\
            A_3 &:= \mbb E[X(\mc K_{h_n}(X^{\top}\beta_0 - Y) - \tau)].
        \end{split}
    \end{equation}
    By the proof of Lemmas \ref{lem:uniform consistency} and the assumption that $n/N = O(1)$, $A_1 = O_p(n^{-1/2})$, while $A_3 = O(h_n^2)$ follows from Lemma \ref{lem:rate_grad_smooth}.
    By Assumption \ref{as:h}, both $A_1$ and $A_3$ are $o_p(h_n)$.
    Thus, $h_n^{-1} A_2 = o_p(1)$.
    Letting $\hat \theta_{h_n} := (\hat \beta_{\text{SD}, \hat W} - \beta_0)/h_n$ for brevity,
    the latter convergence can be written as 
    \begin{equation}\label{prop:rate_PPI:score_op1}
        h_n^{-1} \mbb E [ X \{ \mc K ( (X^{\top}\beta_0 - Y)/h_n + X^{\top}\theta ) - \mc K ( (X^{\top}\beta_0 - Y)/h_n ) \} ]\mid_{\theta = \hat \theta_{h_n}} = o_p(1).
    \end{equation}
    By the fundamental theorem of calculus, the left side can be evaluated as
    \begin{equation}\label{prop:rate_PPI:score_calculus}
        \begin{split}
            &h_n^{-1} \mbb E [ X \{ \mc K ( (X^{\top}\beta_0 - Y)/h_n + X^{\top}\theta ) - \mc K ( (X^{\top}\beta_0 - Y)/h_n ) \} ] \\
            = \ & \mbb E \left[ X X^{\top} \int^{1}_0 h_n^{-1} \mbb E[K((X^{\top}\beta_0 - Y)/h_n + t X^{\top}\theta)|X] dt \right] \theta.
        \end{split}
    \end{equation}
    For the conditional expectation in the integral, the change of variable gives that 
    \begin{equation}
        \begin{split}
            h_n^{-1} \mbb E[K((X^{\top}\beta_0 - Y)/h_n + t X^{\top}\theta)|X] 
            &= \int^{\infty}_{-\infty} K(v) f_{Y|X}(-h_n v + h_n t X^{\top}\theta + X^{\top}\beta_0) dv \notag \\
            &= f_{Y|X} (X^{\top}\beta_0) + R(X, \theta, h_n, t),
        \end{split}
    \end{equation}
    where $R(X, \theta, h_n, t)$ is a function such that $|R(X, \theta, h_n, t)| \leq h_n l_{Y} \kappa_1 + h_n t \| X \| \| \theta \|$ by Assumption \ref{as:density} 
    and \ref{as:kernel}. Plugging this into \eqref{prop:rate_PPI:score_op1} and \eqref{prop:rate_PPI:score_calculus}, we observe that 
    \begin{equation}
        \left\{ \mbb E[f_{Y|X}(X^{\top}\beta_0) X X^{\top}] + \mbb E\left[X X^{\top} \int^1_0 R(X, \theta, h_n, t)dt\right]\mid_{\theta = \hat \theta_{h_n}} \right\} \hat \theta_{h_n} = o_p(1).
    \end{equation}
    Note that the absolute value of each entry of the matrix $\mbb E\left[X X^{\top} \int^1_0 R(X, \theta, h_n, t)dt\right]\mid_{\theta = \hat \theta_{h_n}} $ is bounded 
    by $h_n l_Y \kappa_1 \mbb E[\| X \|^2] + \mbb E[\| X \|^3] \| \hat \beta_{\text{SD}, \hat W} - \beta_0 \|$, which is $o_p(1)$ because of the consistency of $\hat \beta_{\text{SD}, \hat W}$.
    Consequently, it follows that $\{ \mbb E[f_{Y|X} (X^{\top}\beta_0) X X^{\top}] + o_p(1) \} \hat \theta_{h_n} = o_p(1),$ by which
    $\hat \theta_{h_n} = o_p(1)$ holds under Assumption \ref{as:moment}.
    This completes the proof.
\end{proof}

\begin{proof}[Proof of Proposition \ref{prop:asy_normal_LD}]
    The proof is divided into two steps. In the first step, we shows the asymptotic linear expansion.
    \begin{equation}
        \begin{split}
        \sqrt{n} (\hat \beta_{\text{SD}, \hat W} - \beta_0) = \ &\mbb E [f_{Y|X}(X^{\top}\beta_0) X X^{\top}]^{-1}
        \sqrt{n} \mbb P_n X (\tau - \mc K_{h_n} (X^{\top}\beta_0 - Y)) \\
        &- \mbb E [f_{Y|X}(X^{\top}\beta_0) X X^{\top}]^{-1} \hat W  \mbb G_n X (\tau - \mc K_{h_n} (X^{\top}\beta_0 - \hat Y)) \\
        &+  \mbb E [f_{Y|X}(X^{\top}\beta_0) X X^{\top}]^{-1}  \hat W \sqrt{\frac{n}{N}}
        \mbb G_N X (\tau - \mc K_{h_n} (X^{\top}\beta_0 - \hat Y)) + o_p(1).
        \end{split}
    \end{equation}
    Subsequently, we establish the asymptotic normality in the second step.

    \underline{$(i)$ the asymptotic linear expansion}
    This step largely follows that of Theorem 4.2 of \cite{HePanTanZhou2023joe}.
    Define the vector-valued random process $( \Delta (\delta, W))_{\delta \in \mbb R^p, W \in \mbb R^{p \times p}}$ as 
    \begin{equation}
        \Delta (\delta, W) := \hat \Psi_{h_n, W} (\beta_0 + \delta) - \hat \Psi_{h_n, W} (\beta_0) 
        - \mbb E[f_{Y|X}(X^{\top}\beta_0) X X^{\top}] \delta.
    \end{equation}
    By Proposition \ref{prop:rate_PPI} and the assumption on $\hat W$, we can choose a sequence $r_n$ converging to zero 
    such that $\| \hat \beta_{\text{SD}, \hat W} - \beta_0 \| \leq r_n h_n$
    and $\| \hat W - W_0 \| \leq r_n$ with probability approaching one.
    Define $\Theta_n := \{ (\delta, W) : \| \delta \| \leq r_n h_n, \ \text{and} \ \|W  - W_0\| \leq r_n \}$.
    We find an upper bound for $\sup_{(\delta, W) \in \Theta_n} \|\Delta(\delta, W)\|$.
    By the triangle inequality,
    \begin{equation}
        \sup_{(\delta, W) \in \Theta_n} \| \Delta(\delta, W)  \| \leq
        \sup_{(\delta, W) \in \Theta_n} \| \mbb E [\Delta (\delta, W)] \| 
        + \sup_{(\delta, W) \in \Theta_n} \| \Delta(\delta, W) - \mbb E[\Delta (\delta, W)] \|.
    \end{equation}
    We bound the two terms on the right side.

    \underline{$(i$-$a)$ $ \sup_{(\delta, W) \in \Theta_n} \| \mbb E [\Delta (\delta, W)] \|$}
    By the definition of $\Delta(\delta, W)$,
    \begin{equation}\label{prop:asy_normal_LD:Delta_expectation}
        \begin{split}
        \mbb E[\Delta(\delta, W)] = \ &\mbb E[X\{\mc K_{h_n} (X^{\top}(\beta_0 + \delta) - Y)
        - \mc K_{h_n} (X^{\top}\beta_0 - Y)\}] - \mbb E[K_{h_n}(Y - X^{\top}\beta_0) X X^{\top}] \delta \\
        &+ (\mbb E[f_{Y|X}(X^{\top}\beta_0)X X^{\top}] - \mbb E[K_{h_n}(Y-X^{\top}\beta_0) X X^{\top}])\delta.
        \end{split}
    \end{equation}
    By the argument leading to (C.36) in the proof of Theorem 4.2 of \cite{HePanTanZhou2023joe},
    observe that
    \begin{equation}\label{prop:asy_normal_LD:kernel_expansion}
        \sup_{(\delta, W) \in \Theta_n} \| \mbb E[X\{\mc K_{h_n}(X^{\top}(\beta_0 + \delta) - Y) - \mc K_{h_n}(X^{\top}\beta_0 - Y) \}]
        - \mbb E[K_{h_n}(Y - X^{\top}\beta_0) X X^{\top}]\delta \| \lesssim r_n^2 h_n^2.
    \end{equation}
    Meanwhile, we can evaluate $(\mbb E[f_{Y|X}(X^{\top}\beta_0)X X^{\top}] - \mbb E[K_{h_n}(Y-X^{\top}\beta_0) X X^{\top}])\delta$ as 
    \begin{equation}
        \begin{split}
        &\| (\mbb E[f_{Y|X}(X^{\top}\beta_0)X X^{\top}] - \mbb E[K_{h_n}(Y-X^{\top}\beta_0) X X^{\top}])\delta \| \\
        \leq \ &\mbb E \left[ \| X \|^2 \| \delta \| \int^{\infty}_{-\infty} K(u) |f_{Y|X}(X^{\top}\beta_0) - f_{Y|X}(X^{\top}\beta_0 + h_nv)|dv \right]
        \leq h_n \| \delta \| l_Y \kappa_1 \mbb E[\| X \|^2],
        \end{split}
    \end{equation}
    where the last inequality follows from Assumption \ref{as:density} and \ref{as:kernel}.
    Hence, we obtain 
    \begin{equation}\label{prop:asy_normal_LD:density_kernel_diff}
        \sup_{(\delta, W) \in \Theta_n} \| (\mbb E[f_{Y|X}(X^{\top}\beta_0)X X^{\top}] - \mbb E[K_{h_n}(Y-X^{\top}\beta_0) X X^{\top}])\delta \|
        \lesssim r_n h_n^2.
    \end{equation}
    Combining \eqref{prop:asy_normal_LD:Delta_expectation}, \eqref{prop:asy_normal_LD:kernel_expansion}
    and \eqref{prop:asy_normal_LD:density_kernel_diff}, it follows from Assumption \ref{as:h} that
    $ \sup_{(\delta, W) \in \Theta_n} \| \mbb E [\Delta (\delta, W)] \| = o(n^{-1/2}).$

    \underline{$(i$-$b)$ $ \sup_{(\delta, W) \in \Theta_n} \| \Delta(\delta, W) - \mbb E [\Delta (\delta, W)] \|$}
    First, by the definition of $\Delta(\delta, W)$,
    \begin{equation}\label{prop:asy_normal_LD:three_empirical_processes}
        \begin{split}
            \Delta (\delta, W) - \mbb E[\Delta(\delta, W)] = 
            \ &(\mbb P_n - P) [ X \{ \mc K_{h_n}(X^{\top}(\beta_0 + \delta) - Y) - \mc K_{h_n}(X^{\top}\beta_0 - Y)\} ] \\
            &- W \left\{ (\mbb P_n - P) [X\{ \mc K_{h_n} (X^{\top}(\beta_0 + \delta) - \hat Y) - \mc K_{h_n} (X^{\top}\beta_0 - \hat Y) \}] \right. \\
            &- \left. (\mbb P_N - P) [X\{ \mc K_{h_n} (X^{\top}(\beta_0 + \delta) - \hat Y) - \mc K_{h_n} (X^{\top}\beta_0 - \hat Y) \}] \right\}.
        \end{split}
    \end{equation}
    We bound the supremum of each empirical process term on the right side. 
    For the first term, consider a class of function
    $\mc G_n := \{ (Y, X) \mapsto X\{ \mc K_{h_n}(X^{\top}(\beta_0 + \delta) - Y) - \mc K_{h_n}(X^{\top}\beta_0 - Y) \}: (\delta, W) \in \Theta_n \}$.
    Without the loss of generality, we may assume that any $g$ in $\mc G_n$ is a real-valued function.
    By Assumption \ref{as:kernel}, the mean value theorem implies that, for any $(\delta, W) \in \Theta_n$,
    $|X\{ \mc K_{h_n} (X^{\top}(\beta_0 + \delta) - Y) - \mc K_{h_n} (X^{\top}\beta_0 - Y) \}| \leq C_{K} r_n \| X \|^2$
    for some finite constant $C_K$ that only depends on the kernel $K(\cdot)$.
    Hence $G_n(Y, X) := C_{K}r_n \| X \|^2$ is qualified as an envelop function for $\mc G_n$.
    Furthermore, by a similar argument to the proof of Lemma \ref{lem:uniform consistency},
    we can bound the covering number of $\mc G_n$ as
    $N(\ve \| G_n \|_{Q, 2}, \mc G_n, L_2(Q)) \leq C \ve^{-4p - 6} \ (0 < \ve < 1)$, for any 
    probability measure $Q$ where $C$ is a constant independent of $Q$ and $n$.
    It now follows from Theorem 2.14.1 of \cite{VaartWellner2023} and Assumption \ref{as:moment} that 
    $\mbb E[\sup_{g \in \mathcal G_n} \| (\mathbb P_n - P)g \|] \lesssim r_n n^{-1/2}$.
    By Markov's ineqality, we obtain $\sup_{g \in \mathcal G_n} \| (\mbb P_n - P) g \| = o_p(n^{-1/2})$.
    Applying the same argument to the suprema of the second and the third empirical proccess terms on the 
    ride side of \eqref{prop:asy_normal_LD:three_empirical_processes}, we observe that 
    $ \sup_{(\delta, W) \in \Theta_n} \| \Delta(\delta, W) - \mbb E [\Delta (\delta, W)] \| = o_p(n^{-1/2})$.

    Combining substeps $(i$-$a)$ and $(i$-$a)$, we observe that $\sup_{(\delta, W) \in \Theta_n} \| \Delta (\delta, W) \| = o_p(n^{-1/2})$.
    By the definition of $\hat \beta_{\text{SD}, \hat W}$, $\hat W$, and $\Theta_n$, in conjuction with Assumption \ref{as:moment},
    we have established the asymptotic linear expansion.
    \begin{equation}\label{prop:asy_normal_LD:expansion}
        \begin{split}
        \sqrt{n} (\hat \beta_{\text{SD}, \hat W} - \beta_0) = \ &\mbb E [f_{Y|X}(X^{\top}\beta_0) X X^{\top}]^{-1}
        \sqrt{n} \mbb P_n X (\tau - \mc K_{h_n} (X^{\top}\beta_0 - Y)) \\
        &- \mbb E [f_{Y|X}(X^{\top}\beta_0) X X^{\top}]^{-1} \hat W  \mbb G_n X (\tau - \mc K_{h_n} (X^{\top}\beta_0 - \hat Y)) \\
        &+  \mbb E [f_{Y|X}(X^{\top}\beta_0) X X^{\top}]^{-1} \hat W \sqrt{\frac{n}{N}}
        \mbb G_N X (\tau - \mc K_{h_n} (X^{\top}\beta_0 - \hat Y)) + o_p(1).
        \end{split}
    \end{equation}

    \underline{$(ii)$ the asymptotic normality}
    We first evaluate the smoothing bias on the expansion \eqref{prop:asy_normal_LD:expansion}.
    Namely, we show the three convergences: 
    $(ii$-$a)$ $\sqrt{n} \mbb P_n X(\mathbb I(Y - X^{\top}\beta_0 < 0) - \mc K_{h_n} (X^{\top}\beta_0 - Y)) = o_p(1)$,
    $(ii$-$b)$ $\mbb G_n X (\mbb I(\hat Y - X^{\top}\beta_0 < 0) - \mc K_{h_n} (X^{\top}\beta_0 - \hat Y)) = o_p(1)$,
    and
    $(ii$-$c)$ $\mbb G_N X (\mbb I(\hat Y - X^{\top}\beta_0 < 0) - \mc K_{h_n} (X^{\top}\beta_0 - \hat Y)) = o_p(1)$.

    For $(ii$-$a)$, we have the decomposition:
    \begin{equation}
        \begin{split}
        &\sqrt{n} \mbb P_n X(\mathbb I(Y - X^{\top}\beta_0 < 0) - \mc K_{h_n} (X^{\top}\beta_0 - Y)) \\
        = \ &\mbb G_n X(\mathbb I(Y - X^{\top}\beta_0 < 0) - \mc K_{h_n} (X^{\top}\beta_0 - Y)) 
        + \sqrt{n} \mbb E[X(\mathbb I(Y - X^{\top}\beta_0 < 0) - \mc K_{h_n} (X^{\top}\beta_0 - Y))].
        \end{split}
    \end{equation}
    For the empirical process term, a straightforward calculation using
    Markov's inequality gives that, for any $\ve > 0$,
    \begin{equation}
        \mbb P( \| \mbb G_n X(\mathbb I(Y - X^{\top}\beta_0 < 0) - \mc K_{h_n} (X^{\top}\beta_0 - Y)) \| > \ve)
        \leq \mbb E[\| X \|^2 (\mbb I (Y - X^{\top}\beta_0 < 0) - \mc K_{h_n} (X^{\top}\beta_0 - Y))^2]/\ve^2.
    \end{equation}
    By Assumption \ref{as:density}, $\mbb P(Y = X^{\top}\beta_0) = 0$. Hence the dominated convergence theorem yields that the right side of the above display 
    converges to zero, by which $\mbb G_n X(\mathbb I(Y - X^{\top}\beta_0 < 0) - \mc K_{h_n} (X^{\top}\beta_0 - Y)) = o_p(1)$ follows.
    For the expectation term, the proof of Lemma \ref{lem:rate_grad_smooth} shows that 
    $\mbb E[X(\mathbb I(Y - X^{\top}\beta_0 < 0) - \mc K_{h_n} (X^{\top}\beta_0 - Y))] = O(h^2).$
    In light of Assumption \ref{as:h}, 
    $\sqrt{n} \mbb E[X(\mathbb I(Y - X^{\top}\beta_0 < 0) - \mc K_{h_n} (X^{\top}\beta_0 - Y))] = o(1)$ follows.
    Consequently, the convergence $(ii$-$a)$ holds.

    For $(ii$-$b)$, by a similar argument to that using Markov's inequality above, 
    it suffices to show that
    $ \mbb E[\| X \|^2 (\mbb I(\hat Y - X^{\top}\beta_0 < 0) - \mc K_{h_n}(X^{\top}\beta_0 - \hat Y))^2] = o(1)$.
    However, this holds from the dominated convergence theorem in conjunction with Assumption \ref{as:y_hat}.
    We can show $(ii$-$c)$ analogously. 

    By $(ii$-$a)$, $(ii$-$b)$, and $(ii$-$c)$, in conjunction with $\hat W = O_p(1)$ and $n/N = O(1)$,
    the expansion \eqref{prop:asy_normal_LD:expansion} still holds even if 
    the kernel approximation $\mc K_{h_n}$ is replaced with the corresponding 
    indicator functions.
    Furthermore, because the resulting three empirical process terms on the right side of \eqref{prop:asy_normal_LD:expansion}
     are all $O_p(1)$,
    $\hat W \to_p W_0$ and $n/N \to r$ give that
    \begin{equation}
        \begin{split}
        \sqrt{n}(\hat \beta_{PPI, \hat W} - \beta_0) = \ &\mbb E[f_{Y|X}(X^{\top}\beta_0)XX^{\top}]^{-1} 
        \mathbb G_n X (\tau - \mathbb I(Y - X^{\top}\beta_0 < 0)) \\
        & -  \mathbb E[f_{Y|X}(X^{\top}\beta_0) X X^{\top}]^{-1}  W_0 \mathbb G_n X (\tau - \mathbb I (\hat Y - X^{\top}\beta_0 < 0)) \\
        & +  \mathbb E[f_{Y|X}(X^{\top}\beta_0) X X^{\top}]^{-1} W_0 \sqrt{r} \mathbb G_N X (\tau - \mathbb I(\hat Y - X^{\top}\beta_0 < 0)) + o_p(1).
        \end{split}
    \end{equation}
    Because the empirical processes $\mathbb G_n$ and $\mathbb G_N$ are independent, we obtain the desired result.
\end{proof}

\begin{proof}[Proof of Proposition \ref{prop:opt_W}]
    By Assumption \ref{as:moment}, it suffices to show that 
    $\Lambda (W^*) \preceq \Lambda (W)$ for any $W \in \mathbb R^{p \times p}$.
    We fix $W \in \mathbb R^{p \times p}$ arbitrarily.
    Observe that
    $\Lambda(W) = \Lambda_{\text{lab}} + \Omega(W) - (1+r)^{-1} \Lambda_{\text{cov}} \Lambda^{-1}_{\text{unlab}} \Lambda_{\text{cov}}$, where
    $\Omega(W)$ equals
    \begin{equation}
        \left\{ W (1+r)^{1/2} \Lambda_{\text{unlab}}^{1/2} - 
        (1+r)^{-1/2} \Lambda_{\text{cov}} \Lambda_{\text{unlab}}^{-1/2}\right\}
        \left\{ W (1+r)^{1/2} \Lambda_{\text{unlab}}^{1/2} - 
        (1+r)^{-1/2} \Lambda_{\text{cov}} \Lambda_{\text{unlab}}^{-1/2}\right\}^{\top}.
    \end{equation}
    Hence, $W = (1+r)^{-1} \Lambda_{\text{cov}} \Lambda_{\text{unlab}}^{-1}$ gives that $\Omega(W) = 0$, which 
    leads to the desired result.
\end{proof}

\begin{proof}[Proof of Proposition \ref{prop:asy_normal_PTD}]
    Observe that 
    \begin{equation}
        \sqrt{n} (\hat \beta_{\text{SD}, \hat W} - \beta_0) = \sqrt{n} (\hat \beta_{\text{lab}} - \beta_0) 
        - \hat W \left\{ \sqrt{n} (\hat \gamma_{\text{lab}} - \gamma_0) - \sqrt{n} (\hat \gamma_{\text{unlab}} - \gamma_0) \right\}.
    \end{equation}
    We establish the asymptotic linear expansions for $\sqrt{n} (\hat \beta_{\text{lab}} - \beta_0) $,
    $\sqrt{n} (\hat \gamma_{\text{lab}} - \gamma_0)$, and $\sqrt{n} (\hat \gamma_{\text{unlab}} - \gamma_0)$.

    For $\sqrt{n}(\hat \beta_{\text{lab}} - \beta_0)$, note that $\nabla_b \ell_{h_n} (b; y, x) = \psi_{h_n}(b; y, x)$.
    Due to this fact, the convexity of $b \mapsto \ell_{h_n} (b; y, x)$, and 
    Lemma \ref{lem:root} with $W = 0$, we have that 
    $\hat \Psi_{h_n, 0} (\hat \beta_{\text{lab}}) = 0$ with probability approaching one.
    Hence, repeating the argument in the proofs of Propositions \ref{prop:rate_PPI} and \ref{prop:asy_normal_LD},
    we obtain the asymptotic linear expansion:
    \begin{equation}
        \sqrt{n} (\hat \beta_{\text{lab}} - \beta_0) = \mbb E[f_{Y|X} (X^{\top}\beta_0) X X^{\top}]^{-1} \mbb G_n  X (\tau - \mbb I (Y - X^{\top}\beta_0 < 0)) + o_p(1).
    \end{equation}

    By the analogous argument, now employing Assumptions \ref{as:identification_yhat}, \ref{as:density_yhat}, and \ref{as:moment_yhat}, 
    we have the expansions for $\sqrt{n} (\hat \gamma_{\text{lab}} - \gamma_0)$ and 
    $\sqrt{n} (\hat \gamma_{\text{unlab}} - \gamma_0)$ as

    \begin{equation}
        \begin{split}
        \sqrt{n} (\hat \gamma_{\text{lab}} - \gamma_0) &= \mbb E[f_{\hat Y| X} (X^{\top}\gamma_0) X X^{\top}]^{-1} 
        \mbb G_n X (\tau - \mbb I (\hat Y - X^{\top}\gamma_0 < 0)) + o_p(1), \\
        \sqrt{n} (\hat \gamma_{\text{unlab}} - \gamma_0) &= \sqrt{r} \mbb E[f_{\hat Y| X} (X^{\top}\gamma_0) X X^{\top}]^{-1} 
        \mbb G_N X (\tau - \mbb I (\hat Y - X^{\top}\gamma_0 < 0)) + o_p(1).
        \end{split}
    \end{equation}
    By these three expansions in conjunction with the fact that $\hat W \to_p W_0$, we observe that 
    \begin{equation}
        \begin{split}
            \sqrt{n} (\hat \beta_{\text{SD}, \hat W} - \beta_0)
            = \ &\mbb G_n 
            \left\{ \mbb E[f_{Y|X} (X^{\top}\beta_0) X X^{\top}]^{-1}  X (\tau - \mbb I (Y - X^{\top}\beta_0 < 0))   \right. \\
            &\left. - W_0 \mbb E[f_{\hat Y| X} (X^{\top}\gamma_0) X X^{\top}]^{-1} 
            X (\tau - \mbb I (\hat Y - X^{\top}\gamma_0 < 0)) \right\} \\
            &+ \mbb G_N \sqrt{r} W_0 \mbb E[f_{\hat Y| X} (X^{\top}\gamma_0) X X^{\top}]^{-1} 
            X (\tau - \mbb I (\hat Y - X^{\top}\gamma_0 < 0)).
        \end{split}
    \end{equation}
    Because two empirical processes $\mbb G_n$ and $\mbb G_N$ are independent, we obtain the desired result.
\end{proof}

\section{Supporting Lemmas}\label{app:support}

\begin{lemma}\label{lem:root}
    Assume the assumption of Proposition \ref{prop:rate_PPI}.
    Then, if $\hat W = O_p(1)$, there exists $\hat \beta_{\text{SD}, \hat W} \in \mbb R^{p}$ 
    such that $\hat \Psi_{h_n, \hat W} (\hat \beta_{\text{SD}, \hat W}) = 0$ with probability approaching one.
\end{lemma}

\begin{proof}[Proof]
By the proof of Theorem 5.42 of \cite{vaart1998}, 
it is sufficient to show that 
$(i)$ $\sup_{\beta \in \mbb R^p} \| \hat \Psi_{h_n, \hat W} (\beta) - \mbb E[X (\mbb I ( Y - X^{\top} \beta < 0 ) - \tau)] \| \to_p 0$,
$(ii)$ $\mbb E[X (\mbb I (Y - X^{\top} \beta_0 < 0) - \tau)] = 0$, and
$(iii)$ $\beta \mapsto \mbb E[X (\mbb I ( Y - X^{\top} \beta < 0 ) - \tau)]$ is continuously differentiable, and its derivative
matrix is nonsingular at $\beta_0$.
However, $(i)$ is shown in Lemma \ref{lem:uniform consistency} while $(ii)$ can be shown by 
Assumption \ref{as:identification} and Lemma \ref{lem:grad_beta}.
For $(iii)$, by the argument at the end of Lemma \ref{lem:grad_beta}, $\mbb E[X(\mbb I(Y - X^{\top}\beta < 0) - \tau)] = \mbb E[X(F_{Y|X}(X^{\top}\beta) - \tau)]$.
Hence, given Assumption \ref{as:moment} and \ref{as:density},
$\nabla_{\beta} \mbb E [X (\mbb I(Y - X^{\top}\beta < 0) - \tau)] = \mbb E[X X^{\top} f_{Y | X} (X^{\top} \beta)]$,
which is continuous in $\beta$. 
Furthermore, $\mbb E[X X^{\top} f_{Y|X}(X^{\top}\beta_0)]$ is nonsingular by Assumption \ref{as:moment}.
Therefore, $(iii)$ holds, which completes the proof.
\end{proof}

\begin{lemma}\label{lem:uniform consistency}
    Assume the assumption of Proposition \ref{prop:rate_PPI}.
    Then, if $\hat W = O_p(1)$, 
    $\sup_{\beta \in \mbb R^p} \| \hat \Psi_{h_n, \hat W} (\beta) - \mbb E[X (\mbb I ( Y - X^{\top} \beta < 0 ) - \tau)] \| \to_p 0$.
\end{lemma}

\begin{proof}
By the triangle inequality,
\begin{equation}
    \begin{split}
    &\sup_{\beta \in \mbb R^p} \| \hat \Psi_{h_n, \hat W} (\beta) - \mbb E[X (\mbb I ( Y - X^{\top} \beta < 0 ) - \tau)] \| \\
    \leq \ &\sup_{\beta \in \mbb R^p}\left\| \frac{1}{n} \sum_{i = 1}^n \left\{ \psi_{h_n} (Y_i - X_i^{\top} \beta) -
    \mbb E [X( \mbb I (Y - X^{\top} \beta < 0) - \tau)] \right\} \right\| \\
    &+ \| \hat W \| \left\{ \sup_{\beta \in \mbb R^p} \| (\mbb P_n - P) \psi_{h_n} (\hat Y - X^{\top}\beta)  \| 
    + \sup_{\beta \in \mbb R^p} \| (\mbb P_N - P) \psi_{h_n} (\hat Y - X^{\top} \beta) \| \right\}.
    \end{split} \label{score bound}
\end{equation}
It is enough to show that the three supremum terms on the right side converges to zero in probability.

For the first term, by the triangle inequality, this term can be further bounded as
\begin{equation}
    \begin{split}
    &\sup_{\beta \in \mbb R^p}\left\| \frac{1}{n} \sum_{i = 1}^n \left\{ \psi_{h_n} (Y_i - X_i^{\top} \beta) -
    \mbb E [X( \mbb I (Y - X^{\top} \beta < 0))] \right\} \right\| \\
    \leq \ &\sup_{\beta \in \mbb R^p} \| (\mbb P_n - P) \psi_{h_n} (Y - X^{\top}\beta) \| 
    + \sup_{\beta \in \mbb R^p} \| \mbb E [\psi_{h_n} (Y - X^{\top}\beta)] - \mbb E[X (\mbb I(Y - X^{\top} \beta < 0) - \tau)] \|.
    \end{split} \label{label score bound}
\end{equation}
For the first term on the right side, consider a class of function 
$\mc F_n := \{ (Y, X) \mapsto X (\mc K_{h_n} (X^{\top}\beta - Y) - \tau) : \beta \in \mbb R^p \}$.
Without the loss of generality, we may assume that any $f$ in $\mc F_n$ is a real-valued function.
By Assumption \ref{as:kernel}, $F_n := \| X \|$ is a valid envelop function for $\mc F_n$.
We then find the covering number of $\mc F_n$. First, note that
$\{ (Y, X) \mapsto (X^{\top}\beta - Y)/h_n : \beta \in \mbb R^p \}$ is 
a subset of $\{ (Y, X) \mapsto a_Y Y + a_X^{\top} X: (a_Y, a_X) \in \mbb R^{p + 1}\}$,
which is a $p+1$ dimensional vector space. Hence, both clasess are VC-subgraph with their dimension bounded by $p+2$
by Lemma 2.6.16 of \cite{VaartWellner2023}.
Subsequently, applying Lemma 9.0 $(viii)$, $(v)$, and $(vi)$ of \cite{Kosorok2008}, in turn,
observe that $\mc F_n$ is also VC-subgraph with its dimension bounded by $2p + 3$.
Consequently, Theorem 2.6.7 of \cite{VaartWellner2023} yields the bound on the covering number as $N(\ve \| F_n \|_{Q, 2}, \mc F_n, L_2(Q)) \leq K \ve^{-4p-6} \ (0 < \ve < 1)$,
for any probability measure $Q$, where $K$ is a constant independent of $Q$ and $n$, and $\| \cdot \|_{Q, 2}$
is the $L_2$-norm with respect to a measure $Q$.
Refer to, for example, Definition 2.1.5 of \cite{VaartWellner2023} for the definition of the covering number.
It now follows from Theorem 2.14.1 of \cite{VaartWellner2023} and Assumption \ref{as:moment} that
$\mbb E \left[ \sup_{\beta \in \mbb R^{p}} \| (\mbb P_n - P) \psi_{h_n} (Y - X^{\top}\beta) \| \right] \lesssim n^{-1/2}$.
By Markov's inequality, $\sup_{\beta \in \mbb R^{p}} \| (\mbb P_n - P) \psi_{h_n} (Y - X^{\top}\beta) \| \to_p 0$ follows.

For the second term on the right side of \eqref{label score bound}, observe that 
\begin{equation}
    \sup_{\beta \in \mbb R^p} \| \mbb E [\psi_{h_n} (Y - X^{\top}\beta)] - \mbb E[X (\mbb I(Y - X^{\top} \beta < 0) - \tau)] \|
    = \sup_{\beta \in \mbb R^p} \| \mbb E[X (\mc K_{h_n}(X^{\top}\beta - Y) - \mbb I(Y - X^{\top}\beta < 0))]\|.
\end{equation}
For the expectation in the supremum, we have
\begin{equation}
     \mbb E[X (\mc K_{h_n}(X^{\top}\beta - Y) - \mbb I(Y - X^{\top}\beta < 0))] 
    = \mbb E [X \mbb E[\mc K_{h_n}(X^{\top}\beta - Y) - \mbb I(Y - X^{\top}\beta < 0) | X]]. \label{diff kernel ind}
\end{equation}
The conditional expectation can be evaluated as
\begin{equation}
    \begin{split}
        &\mbb E[\mc K_{h_n}(X^{\top}\beta - Y) - \mbb I(Y - X^{\top}\beta < 0) | X] \\
        = \ &\int^{\infty}_{-\infty} \mc K \left( \frac{X^{\top}\beta - u}{h_n} \right) f_{Y | X}(u) du
         - F_{Y | X} (X^{\top}\beta) \\
        = \ &\int^{\infty}_{-\infty} \frac{1}{h_n} K \left( \frac{X^{\top}\beta - u}{h_n} \right) F_{Y | X} (u) du 
        - F_{Y | X} (X^{\top} \beta) \\
        = \ &\int^{\infty}_{-\infty} K(v) \left\{ F_{Y | X} (X^{\top}\beta - h_nv) - F_{Y | X} (X^{\top}\beta) \right\} dv,
    \end{split}
\end{equation}
where the second equality follows from the integration by parts,
and the third one is from the change of variable. By Assumption \ref{as:kernel} and \ref{as:density}, the absolute value of the term
in the above display is bounded by $h_n f_{\sup} \kappa_1$. In view of \eqref{diff kernel ind}, observe that
\begin{equation}
    \sup_{\beta \in \mbb R^p} \| \mbb E[X (\mc K_{h_n}(X^{\top}\beta - Y) - \mbb I(Y - X^{\top}\beta < 0))]\|
    \leq h_n f_{\sup} \kappa_1 \mbb E[\| X \|],
\end{equation}
which goes to zero as $n \to \infty$ by Assumption \ref{as:h}. This shows the convergence of the second term on 
the right side of \eqref{label score bound}, and consequently the convergence on the first term on the right side of \eqref{score bound}.

The convergence of the second term on the right side of \eqref{score bound} can be shown in the same way as $\sup_{\beta \in \mbb R^{p}} \| (\mbb P_n - P) \psi_{h_n} (Y - X^{\top}\beta) \| \to_p 0$ above.
This completes the proof.
\end{proof}

\begin{lemma}\label{lem:grad_beta}
    Assume Assumptions \ref{as:identification}, \ref{as:density}, and \ref{as:moment}. Then $\nabla_{\beta} \mbb E[\rho_{\tau}(Y - X^{\top}\beta)-\rho_{\tau}(Y)] = \mbb E[X (\mbb I(Y - X^{\top}\beta < 0) - \tau)]$.
\end{lemma}

\begin{proof}
    By Assumption~\ref{as:density}, $Y \neq X^{\top} \beta$ almost surely. Hence, the map $\beta \to \rho_{\tau} (Y - X^{\top} \beta) - \rho_{\tau}(Y)$ is differentiable almost surely, with its gradient $\{ \mathbb I(Y - X^{\top} \beta < 0) - \tau \} X$.
    This gradient is bounded by $\| X \|$ uniformly over $\beta \in \mathbb{R}^p$. Consequently, the result follows from the dominated convergence theorem together with Assumption \ref{as:moment}.
\end{proof}

\begin{lemma}\label{lem:convex}
    Let $M(\beta): \mbb R^p \mapsto \mbb R$ be a continuously differentiable convex function.
    Suppose that $M(\beta)$ is uniquely minimized at some $\beta^*$.
    Then, if $D(\beta)$ is a gradient of $M (\beta)$ with respect to $\beta$,
    $\inf_{\beta: \| \beta- \beta^* \| \geq \ve } \|D (\beta) \| > 0$ for any $\ve > 0$.
\end{lemma}

\begin{proof}
    For any unit vector $v \in \mbb S^{p-1} := \{v \in \mbb R^{p} : \| v \| = 1 \}$, define a function 
    \begin{equation}
        g_v(t) := M(\beta^* + t v), \quad t \in \mbb R.
    \end{equation}
    By convexity of $M$, this $t \mapsto g_{v} (t)$ is also convex.
    Furthermore, $t = 0$ is the unique minimizer of this $g_{v} (t)$ for any $v \in \mbb S^{p-1}$
    since $\beta^*$ is the unique minimizer of $\beta \mapsto M(\beta)$.
    Hence, $g_v'(0) = v^{\top} D(\beta^*) = 0$, and $g'_v(t) = v^{\top} D(\beta^* + tv)$ satisfies
    \begin{equation}
        \quad g'_v(t) > 0 \ \text{when} \ t > 0. \label{lab:g_condition}
    \end{equation}

    Fix $\ve > 0$. Consider a function $h_{\ve} (v) := g'_v(\ve) = v^{\top} D(\beta^* + t v)$
    for $v \in \mbb S^{p-1}$. By \eqref{lab:g_condition}, $h_{\ve}(v) > 0$ 
    for any $v \in \mbb S^{p-1}$. By continuity of $v \mapsto h_{\ve} (v)$ and compactness of $\mbb S^{p-1}$, $c(\ve) := \inf_{v \in \mbb S^{p-1}} h_{\ve} (v) > 0$.
    Now, for any $\beta \in \mbb R^p$ such that $\| \beta - \beta^* \| \geq \ve$, define $r := \| \beta - \beta^* \|$, and
    $v := (\beta - \beta^*)/\| \beta - \beta^* \|$.
    Then $\beta = \beta^* + r v$ and $g'_v (r) \geq g_v'(\ve) \geq c(\ve) > 0$. Furthermore, $g'_v(r) = v^{\top} D(\beta) \leq \| v \| \| D(\beta) \| = \| D(\beta) \|$ where the inequality follows from the Cauchy-Shcwarz inequality.
    It follows that $\| D(\beta) \| \geq c(\ve)$ and the proof is complete.
\end{proof}

\begin{lemma}\label{lem:rate_grad_smooth}
    Assume the assumption of Proposition \ref{prop:rate_PPI}.
    Then $\mbb E[X(\mc K_{h_n} (X^{\top}\beta_0 - Y) - \tau)] = O(h_n^2)$.
\end{lemma}

\begin{proof}
    Observe that 
    \begin{equation}
        \mbb E[X(\mc K_{h_n} (X^{\top}\beta_0 - Y) - \tau)] = \mbb E[X(\mc K_{h_n} (X^{\top}\beta_0 - Y) - \mbb I(Y - X^{\top}\beta_0 < 0))] +
        \mbb E[X(\mbb I(Y - X^{\top}\beta_0 < 0) - \tau)].
    \end{equation}
    By Assumption \ref{as:identification} and Lemma \ref{lem:grad_beta}, the second term on the right side is zero.
    Hence, we only need to show that the first term is $O(h_n^2)$.
    Observe that $\mbb E[X(\mc K_{h_n} (X^{\top}\beta_0 - Y) - \mbb I(Y - X^{\top}\beta_0 < 0))] = \mbb E[X \mbb E[\mc K_{h_n} (X^{\top}\beta_0 - Y) - \mbb I(Y - X^{\top}\beta_0 < 0)|X]]$.
    We can evaluate the conditional expectation as 
    \begin{equation}
        \begin{split}
            &\mbb E[\mc K_{h_n}(X^{\top}\beta - Y) - \mbb I(Y - X^{\top}\beta_0 < 0)|X] \notag \\
            = \ &h_n \int^{\infty}_{-\infty} \{ \mc K(v) - \mbb I(v>0) \} f_{Y|X}(X^{\top}\beta_0 - h_n v) dv \notag \\
            = \ &h_n f_{Y|X} (X^{\top}\beta_0) \int^{\infty}_{-\infty} \{ \mc K(v) - \mbb I(v > 0) \} dv \notag \\
            &+ h_n \int^{\infty}_{-\infty} \{ \mc K(v) - \mbb I(v > 0) \} \{ f_{Y|X} (X^{\top}\beta_0 - h_nv) - f_{Y|X} (X^{\top}\beta_0) \} dv,
        \end{split}
    \end{equation}
    where the first equality is due to the change of variable. 
    Note that all the integrals in the above display are finite by the fact that 
    $\int^{\infty}_{-\infty} |\mc K(v) - \mbb I(v > 0)| dv < \infty$ from the proof of Lemma 5 of \cite{Takeishi2023sjs}, 
    in conjunction with Assumptions \ref{as:density} and \ref{as:kernel}.
    Because $v \mapsto \mc K(v) - \mbb I(v > 0)$ is an odd function, we have that 
    $\int^{\infty}_{-\infty} \{ \mc K(v) - \mbb I(v > 0) \} dv = 0$. It then follows from Assumption \ref{as:density} that 
    \begin{equation}
        |\mbb E[\mc K_{h_n}(X^{\top}\beta - Y) - \mbb I(Y - X^{\top}\beta_0 < 0)|X]|  \leq h_n^2 l_{Y} \int^{\infty}_{-\infty} |\mc K(v) - \mbb I(v > 0)| |v| dv.
    \end{equation}
    To show that integral on the right side is finite, we follow the argument in the proof of Lemma 5 of \cite{Takeishi2023sjs}. Namely,
    \begin{equation}
        \begin{split}
        \int^{\infty}_{-\infty} |\mc K(v) - \mbb I(v > 0)| |v| dv = \ &-\int^{0}_{-\infty} v \mc K(v) dv + \int^{\infty}_{0} v (1 - \mc K(v)) dv \notag \\
        = \ &-[v^2 \mc K(v)]^{0}_{-\infty} + [v^2(1 - \mc K(v))]^{\infty}_0 + \int^{\infty}_{-\infty} v^2 \mc K(v) dv,
        \end{split}
    \end{equation}
    where the second equality follows from integration by parts.
    By L'H\^opital's rule and Assumption \ref{as:kernel}, 
    the last three terms in the above display are finite.
    This completes the proof.
\end{proof}

\section{Undercoverage for the coefficient on $X_1$}\label{app:undercoverage}
As suggested in the main text, the coefficient on $X_1$ exhibits undercoverage at smaller sample sizes, for ``SD-OPT'' intervals at $\tau=0.9$. To examine this behavior more directly, we conduct a focused diagnostic in the setting $\tau=0.9$ and $(\sigma_Z^2,\sigma_u^2)=(1,1)$. We modify the heteroskedastic term in model \eqref{heteroskedastic} from $(1+0.5X_1)$ to $(1+c_1X_1)$ and vary the loading $c_1$.

Figure~\ref{fig:undercoverage-x1-c1-n} reports the resulting empirical coverage for the coefficient on $X_1$. The left panel shows that coverage deteriorates as the heteroskedasticity loading on $X_1$ increases. The right panel fixes $c_1=0.5$ and varies the labeled sample size, showing that the undercoverage generally attenuates as the sample size increases, despite some finite-sample fluctuation. Thus, the coefficient-specific undercoverage highlighted in the main text is most pronounced when heteroskedasticity is strongly tied to $X_1$ and becomes less severe in larger samples.

\begin{figure}[tb]
  \centering
  \begin{minipage}[t]{0.485\textwidth}
    \centering
    \includegraphics[width=\linewidth]{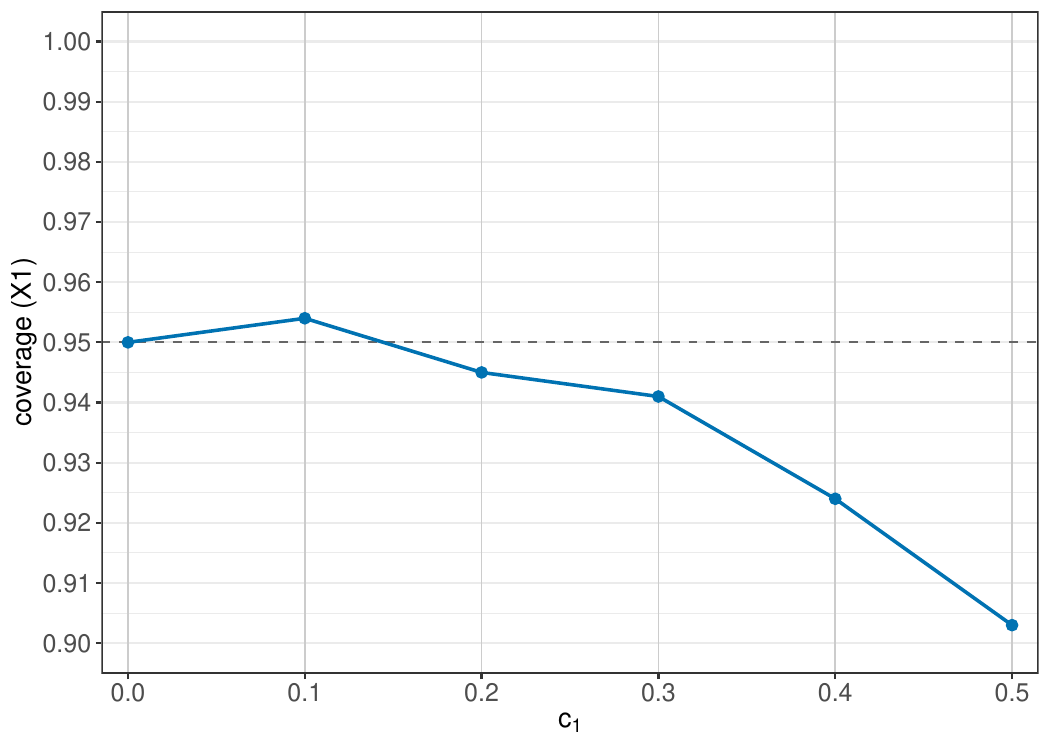}\\[-1ex]
    \textbf{(a)} Coverage of $X_1$ versus heteroskedasticity level $c_1$.
  \end{minipage}
  \hfill
  \begin{minipage}[t]{0.485\textwidth}
    \centering
    \includegraphics[width=\linewidth]{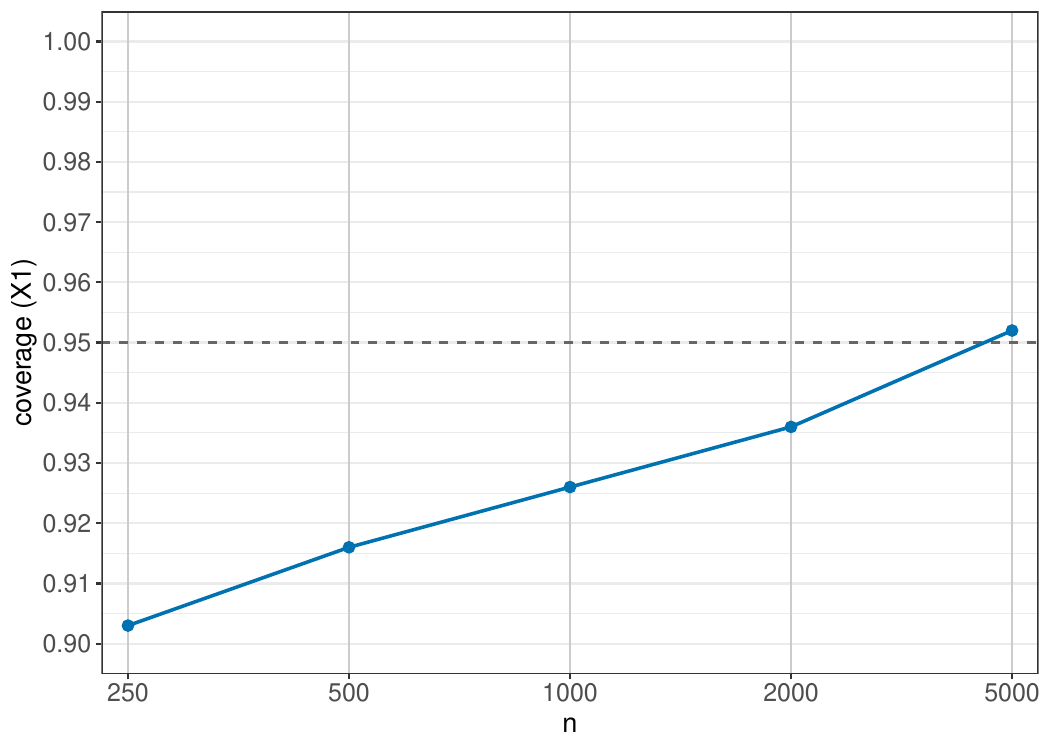}\\[-1ex]
    \textbf{(b)} Coverage of $X_1$ versus sample size ($c_1=0.5$).
  \end{minipage}
  \caption{Undercoverage behavior for the coefficient on $X_1$ at $\tau=0.9$ in the large-variance setting: dependence on (a) heteroskedasticity loading $c_1$ and (b) sample size.}
  \label{fig:undercoverage-x1-c1-n}
\end{figure}

\section{Sensitivity analysis with respect to $h_n$}\label{app:sensitivity}

We conduct a brief sensitivity analysis of the SD-CSE and PTD-CSE with respect to the bandwidth $h_n$.
We use the same simulation setting as in Section~\ref{sec:simulation}, but focus on the case $(\sigma^2_Z, \sigma^2_u) = (1, 1)$,
$\tau = 0.5$, and $n = 1000$.
We examine how the empirical coverage and the length of the confidence interval vary with $h_n$.
Note that, under this setting, our default choice for $h_n$ is $\{ (p + \log n)/n \}^{2/5} \approx 0.16$.
We consider five values of $h_n$: 0.06, 0.11, 0.16, 0.21, and 0.26.
Figure \ref{fig:ci-coverage-length-h} shows how empirical coverage and confidence-interval length, averaged across the coefficients on $X_1$, $X_2$, and $X_3$, vary across values of $h_n$ for ``LAB,''
``SD-OPT,'' and ``PTD-OPT.''
The performance of all three methods is not strongly affected by the choice of $h_n$, and ``SD-OPT'' and ``PTD-OPT'' consistently dominate ``LAB.''
Overall, these findings align with the insensitivity of convolution-smoothed quantile regression to the choice of $h_n$, as observed in the numerical studies of \citet{HePanTanZhou2023joe}.

\begin{figure}[tb]
  \centering
  \begin{minipage}[t]{0.485\textwidth}
    \centering
    \includegraphics[width=\linewidth]{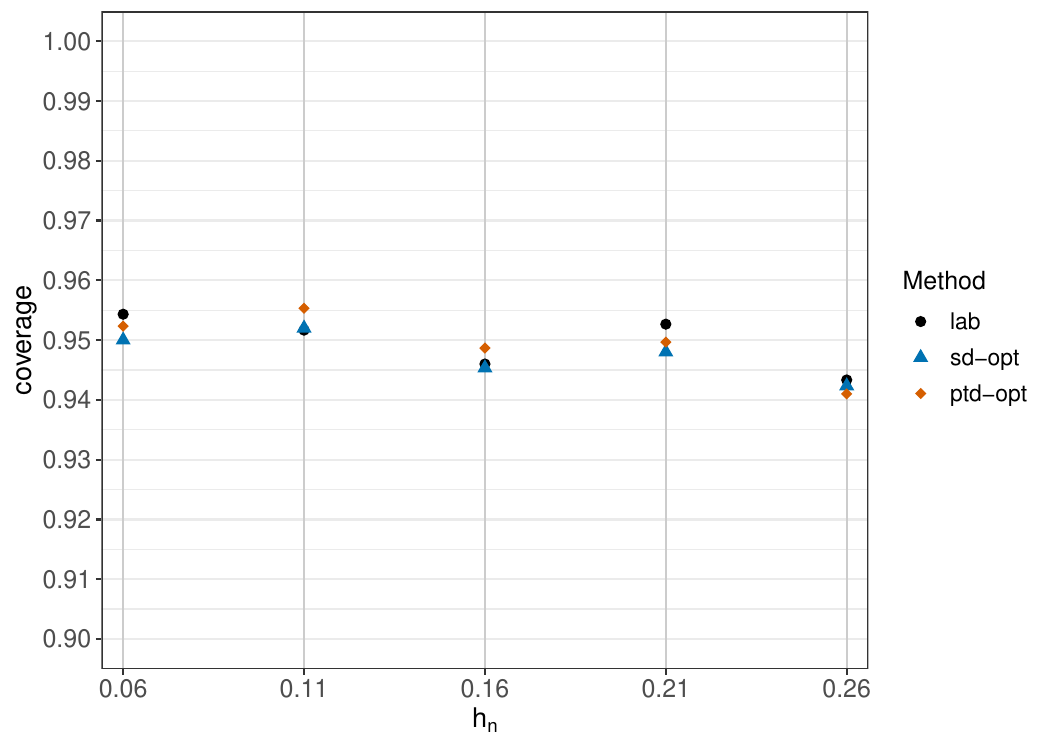}\\[-1ex]
    \textbf{(a)} Empirical coverage
  \end{minipage}
  \hfill
  \begin{minipage}[t]{0.485\textwidth} \centering\includegraphics[width=\linewidth]{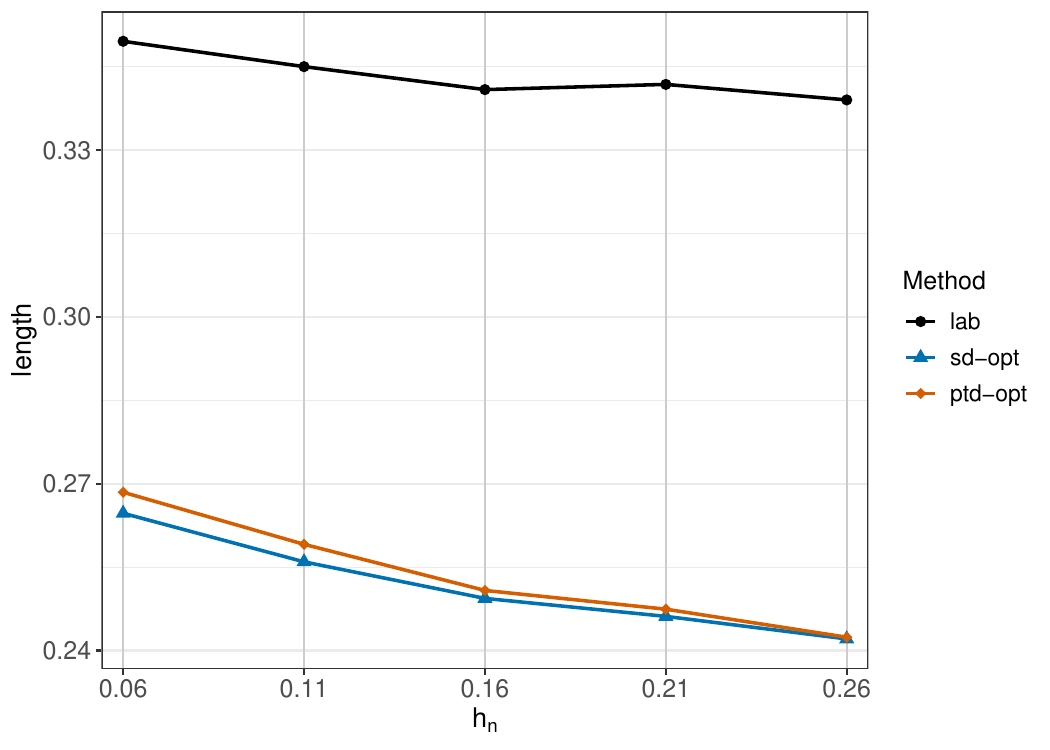}\\[-1ex]
    \textbf{(b)} CI length
  \end{minipage}
  \caption{Sensitivity to \(h_n\), averaged across coefficients \(X_1\), \(X_2\), and \(X_3\) (\(\tau=0.5\), large setting): (a) empirical coverage and (b) CI length.}
  \label{fig:ci-coverage-length-h}
\end{figure}

\section{Calibration of $\hat Y$}\label{app:calibration}

In Section \ref{sec:theory:sd}, we discuss calibrating the surrogate $\hat Y$ by fitting a quantile regression of $Y$ on $\hat Y$ at the target quantile level $\tau$. We then construct the calibrated surrogate $\hat \eta_1 + \hat \eta_2 \hat Y$, where $(\hat \eta_1, \hat \eta_2)$ denotes the resulting coefficient estimates.
Let $\eta^* = (\eta_{1}^*, \eta_{2}^*)$ denote the probability limit of $\hat \eta = (\hat \eta_1, \hat \eta_2)$.
In this section, we show that estimation of $\eta^*$ by $\hat \eta$ does not affect the asymptotic variance of the SD estimator with the estimated optimal weight matrix. In other words, we may treat $\eta^*$ as if it were known.
We further demonstrate the validity of the proposed calibration step through a preliminary simulation study.

Define the smoothed score with the calibrated surrogate as
\begin{equation}
    \hat \Psi_{h_n, \hat W} (\beta; \hat \eta) 
    := \frac{1}{n} \sum_{i = 1}^n \psi_{h_n} (\beta; Y_i, X_i)
    - \hat W \left( 
    \frac{1}{n} \sum_{i = 1}^n \psi_{h_n} (\beta; \hat Y_i(\hat \eta), X_i) -
    \frac{1}{N} \sum_{i = n+1}^{n+N} \psi_{h_n} (\beta; 
    \hat Y_i(\hat \eta), X_i)
    \right),
\end{equation}
where $\hat Y_i (\eta) = \eta_1 + \eta_2 \hat Y_i$ for $\eta = (\eta_1, \eta_2)$.

\begin{proposition}\label{prop:asy_normal_LD_calibrated}
    Assume Assumptions \ref{as:identification}, \ref{as:density}, \ref{as:moment}, \ref{as:kernel} and \ref{as:h}. 
    Furthermore, suppose that $\eta_1^* + \eta_2^* \hat Y \neq X^{\top} \beta_0$ almost surely, that $\hat W \to_p W_0$ for some matrix $W_0$,  that $\hat Y$ has a finite fourth moment, that $n/N \to r$, and that $\sqrt{n} (\hat \eta - \eta^*) = O_p(1)$.
    Then there exists $\hat \beta_{\text{SD}, \hat W} \in \mathbb R^p$ such that $\hat \Psi_{h_n, \hat W} (\hat \beta_{\text{SD}, \hat W}; \hat \eta) = 0$ holds with probability approaching one. Moreover, for any such $\hat \beta_{\text{SD}, \hat W}$,
    \[
        \sqrt{n} (\hat \beta_{\text{SD}, \hat W} - \beta_0)
        \to_d N(0, H^{-1} \Lambda(W_0) H^{-1}), 
    \]
    where $H$ and $\Lambda (W_0)$ are defined as in Proposition \ref{prop:asy_normal_LD}, except that $\hat Y$ is replaced by $\hat Y(\eta^*)$.
\end{proposition}

\begin{proof}
    By a straightforward adaptation of the proof of Lemma \ref{lem:uniform consistency}, with $\hat \eta_1 + \hat \eta_2 \hat Y$ in place of $\hat Y$, we can show the existence part and that
    $\hat \beta_{\text{SD}, \hat W} - \beta_0 = o_p(h_n)$
    by repeating the arguments used in the proofs of Proposition \ref{prop:rate_PPI} and Lemma \ref{lem:root}. For the asymptotic normality, by an argument analogous to that leading to \eqref{prop:asy_normal_LD:expansion}, we obtain the following asymptotic expansion:
    \begin{equation}\label{prop:asy_normal_LD_calibration:expansion}
        \begin{split}
        \sqrt{n} (\hat \beta_{\text{SD}, \hat W} - \beta_0) = \ &\mbb E [f_{Y|X}(X^{\top}\beta_0) X X^{\top}]^{-1}
        \sqrt{n} \mbb P_n X (\tau - \mc K_{h_n} (X^{\top}\beta_0 - Y)) \\
        &- \mbb E [f_{Y|X}(X^{\top}\beta_0) X X^{\top}]^{-1} \hat W  \frac{1}{\sqrt{n}} \sum_{i = 1}^n X_i \{ \tau - \mc K_{h_n} (X_i^{\top}\beta_0 - \hat Y_i(\hat \eta)) \} \\
        &+  \mbb E [f_{Y|X}(X^{\top}\beta_0) X X^{\top}]^{-1} \hat W \sqrt{\frac{n}{N}}
        \frac{1}{\sqrt{N}} \sum_{i = n+1}^{n+N} X_i \{ \tau - \mc K_{h_n} (X_i^{\top}\beta_0 - \hat Y_i (\hat \eta)) \} + o_p(1).
        \end{split}
    \end{equation}
    It follows from Assumption \ref{as:h} and $\sqrt{n} (\hat \eta - \eta^*) = O_p(1)$ that $\hat \eta - \eta^* = h_n o_p(1)$.
    Hence, $\hat \eta \in H_n$ with probability approaching one, where $H_n := \{ \eta : \| \hat \eta - \eta^* \| \leq h_n s_n  \}$ for some sequence $s_n$ converging to zero.
   Then, \eqref{prop:asy_normal_LD_calibration:expansion} can be further simplified as
    \begin{equation}\label{prop:asy_normal_LD_calibration:expansion_rem}
        \begin{split}
        \sqrt{n} (\hat \beta_{\text{SD}, \hat W} - \beta_0) = \ &\mbb E [f_{Y|X}(X^{\top}\beta_0) X X^{\top}]^{-1}
        \sqrt{n} \mbb P_n X (\tau - \mc K_{h_n} (X^{\top}\beta_0 - Y)) \\
        &- \mbb E [f_{Y|X}(X^{\top}\beta_0) X X^{\top}]^{-1} \hat W  \mbb G_n X \{ \tau - \mc K_{h_n} (X^{\top}\beta_0 - \hat Y(\eta^*))\} \\
        &+  \mbb E [f_{Y|X}(X^{\top}\beta_0) X X^{\top}]^{-1} \hat W \sqrt{\frac{n}{N}}
        \mbb G_N X \{ \tau - \mc K_{h_n} (X^{\top}\beta_0 - \hat Y(\eta^*))\} 
        + R_{n, N}
        + o_p(1),
        \end{split}
    \end{equation}
    where, with probability approaching one, the remainder term $R_{n, N}$ can be bounded as
    \begin{equation}
        \begin{split}
        |R_{n, N}| \lesssim \ &\sup_{\eta \in H_n} |\mathbb G_n X \{ \mathcal{K}_{h_n}(X^{\top}\beta_0 - \hat Y(\eta^*)) - \mathcal{K}_{h_n} (X^{\top} \beta_0 - \hat Y(\eta)) \} | \\
        &+ \sup_{\eta \in H_n} |\mathbb G_N X \{ \mathcal{K}_{h_n}(X^{\top}\beta_0 - \hat Y(\eta^*)) - \mathcal{K}_{h_n} (X^{\top} \beta_0 - \hat Y(\eta) \} |.
        \end{split}
    \end{equation}
    Using an argument similar to that used for the supremum of the empirical processes in part $(i)$-$(b)$ of the proof of Proposition \ref{prop:asy_normal_LD}, the expectation of the right-hand side of the above display is bounded, up to a constant, by $s_n \sqrt{\mathbb E[|X|^2 (1 + |\hat Y|^2)]}$ by Theorem 2.14.1 of \cite{VaartWellner2023}. By Markov's inequality, $R_{n, N} = o_{p} (1)$.
    Based on \eqref{prop:asy_normal_LD_calibration:expansion_rem}, the desired result follows by repeating the argument in part $(ii)$ of the proof of Proposition \ref{prop:asy_normal_LD}.
\end{proof}

We now assess the performance of SD-CSE with the calibrated surrogate through a simulation study.
We use the same simulation setting as in Section~\ref{sec:simulation}, but redefine the original surrogate as $\hat Y = \hat f(X_i, Z_i) - 2$ to create a setting in which it underestimates $Y$.
We focus on the $0.9$-th quantile.
Under this setup, $\mathbb{I} (\hat Y_i - X_i^{\top} \hat \beta_{\text{lab}}(0.9) < 0)$ is almost always close to one,  which can make the plug-in estimator of $\Lambda_{\text{unlab}}$ nearly singular. We then apply the proposed calibration method by fitting a linear quantile regression of $Y$ on $\hat Y$ at $\tau =0.9$.
Figure~\ref{fig:sim-coverage-calibration} reports empirical coverages of the ``SD-OPT'' for the slope coefficients under different choices of $n$ and $(\sigma^2_Z, \sigma^2_u)$ at $\tau = 0.9$.
For comparison, we also include the corresponding results for ``LAB.''
The coverages probabilities follow a pattern similar to that observed in Section~\ref{sec:simulation}: they generally attain the nominal level with occasional undercoverage for the $X_1$ coefficient, especially at smaller sample sizes.
Figure~\ref{fig:sim-length-calibration} compares the interval for ``LAB'' and ``SD-OPT.''
As in Section~\ref{sec:simulation}, ``SD-OPT'' yields intervals that are shorter than, or at least comparable in length to, those of LAB. The reduction in interval length is more pronounced when the variance of $Z$ is larger.

\begin{figure}[t]
    \centering
    \includegraphics[width=0.95\textwidth]{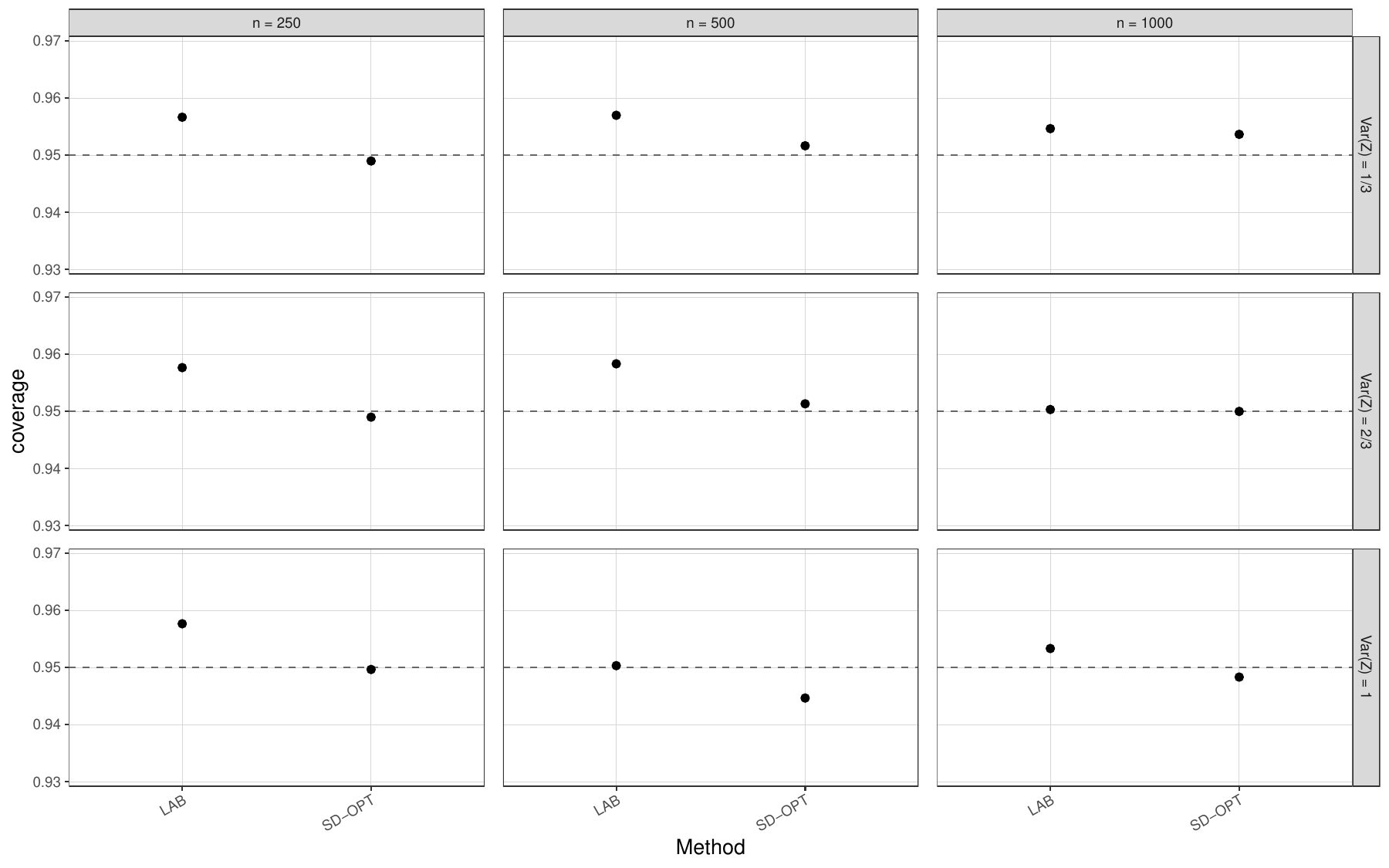}
    \caption{Empirical coverage for 95\% confidence intervals for ``LAB'' and ``SD-OPT'' at $\tau = 0.9$. Each point shows the coverage averaged over the three slope coefficients ($X_1$, $X_2$, or $X_3$; intercept excluded), across sample sizes and variance levels.}
    \label{fig:sim-coverage-calibration}
\end{figure}

\begin{figure}[t]
    \centering
    \includegraphics[width=0.95\textwidth,height=0.78\textheight,keepaspectratio]{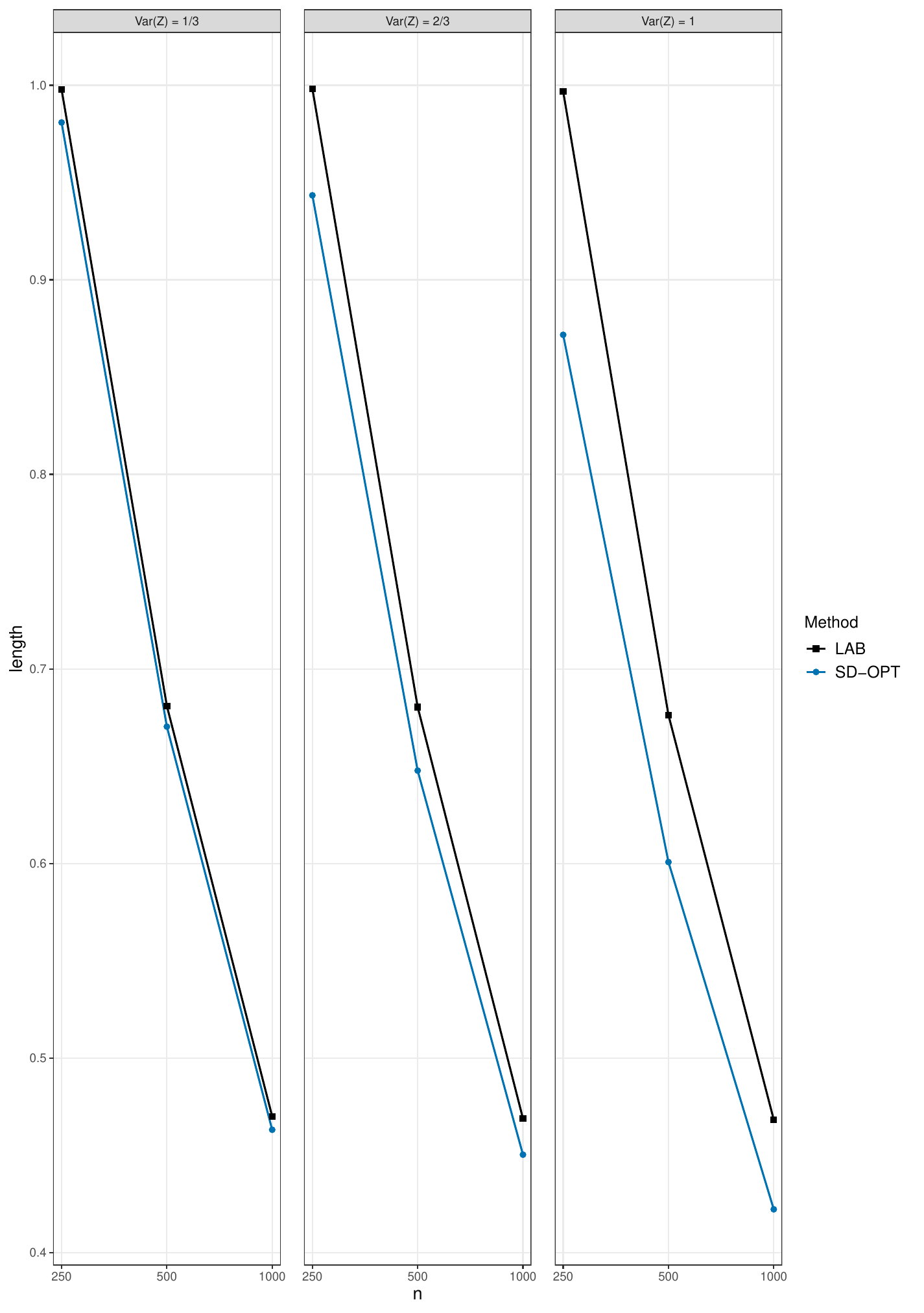}
    \caption{Confidence interval length averaged across the coefficients $X_1$, $X_2$, and $X_3$ at $\tau = 0.9$, over 1000 replications.}
    \label{fig:sim-length-calibration}
\end{figure}

\newpage

\bibliographystyle{asa} 
\bibliography{refs}

\end{document}